\documentclass[11pt,letterpaper]{article}

\usepackage[letterpaper,margin=1in]{geometry}
\usepackage{setspace}
\usepackage[T1]{fontenc}
\usepackage{newtxtext}

\usepackage{amsmath}
\usepackage{amsthm}
\usepackage{amssymb}
\usepackage{mathtools}
\usepackage{microtype}
\usepackage{etoolbox}

\usepackage{authblk}
\usepackage{graphicx}
\usepackage{booktabs}
\usepackage{makecell}
\usepackage{diagbox}
\usepackage[font=normalsize,labelfont=bf]{caption}
\usepackage{algorithm}
\usepackage{float}
\usepackage{placeins}
\usepackage[noend]{algpseudocode}
\algrenewcommand\algorithmicrequire{\textbf{Input:}}
\algrenewcommand\algorithmicensure{\textbf{Output:}}

\AtBeginEnvironment{algorithmic}{\normalsize}

\usepackage{enumitem}
\setlist[itemize]{leftmargin=1.8em,itemsep=0.12em,topsep=0.30em}
\setlist[enumerate]{leftmargin=1.8em,itemsep=0.12em,topsep=0.30em}

\usepackage[numbers,sort&compress]{natbib}
\usepackage[colorlinks=true,citecolor=blue,linkcolor=blue,urlcolor=blue]{hyperref}
\usepackage[nameinlink,noabbrev]{cleveref}

\makeatletter
\providecommand*{\theHALG@line}{}
\renewcommand*{\theHALG@line}{\thealgorithm.\arabic{ALG@line}}
\makeatother

\theoremstyle{plain}
\newtheorem{theorem}{Theorem}[section]
\newtheorem{lemma}[theorem]{Lemma}
\newtheorem{proposition}[theorem]{Proposition}

\newtheorem{observation}[theorem]{Observation}

\theoremstyle{definition}
\newtheorem{definition}[theorem]{Definition}

\makeatletter
\renewenvironment{abstract}
{%
    \normalsize
    \begin{center}
        \bfseries \abstractname
    \end{center}
    \quotation
}
{%
    \endquotation
}
\makeatother

\newcommand{\supp}{\operatorname{supp}}

\title{\textbf{%
        Beyond 1/3-Approximation: Random Greedy for Constrained \(\boldsymbol{k}\)-Submodular Maximization}%
        \textsuperscript{\(\dagger\)}}

\author[1]{Siyuan Chen}
\author[3,4]{Shengminjie Chen}
\author[1,2]{Suixiang Gao}
\author[1]{Zheyu Jiang}
\author[5,6]{Chenhao Wang}
\author[1]{Wenguo Yang}

\affil[1]{School of Mathematical Sciences,
    University of Chinese Academy of Sciences,
    Beijing 100049, China}

\affil[2]{Zhongguancun Laboratory,
    Beijing 100000, China}

\affil[3]{State Key Lab of Processors,
    Institute of Computing Technology,
    Chinese Academy of Sciences,
    Beijing 100190, China}

\affil[4]{School of Computer Science and Technology,
    University of Chinese Academy of Sciences,
    Beijing 100049, China}

\affil[5]{Beijing Normal University-Zhuhai,
    Zhuhai, China}

\affil[6]{Beijing Normal-Hong Kong Baptist University,
    Zhuhai, China}

\date{\vspace{-1.5em}
    \small\texttt{
        chensiyuan221@mails.ucas.ac.cn,
        csmj@ict.ac.cn,
        sxgao@ucas.ac.cn,
        jiangzheyu24@mails.ucas.ac.cn,
        chenhwang@bnu.edu.cn,
        yangwg@ucas.ac.cn
    }
}

\begin{document}

\hypersetup{pageanchor=false}

\begin{titlepage}

    \maketitle

    \begingroup
    \renewcommand{\thefootnote}{\ensuremath{\dagger}}
    \footnotetext{The authors are listed in alphabetical order by surname.}
    \endgroup

    \thispagestyle{empty}

    \vspace{1.5em}

    \begin{abstract}
        \(k\)-submodularity generalizes submodularity by allowing each selected element to be assigned one of \(k\) labels. We study the maximization of a non-monotone \(k\)-submodular function, where \(k\ge2\), under matroid and knapsack constraints. We develop a proportional Top-2 randomized framework that achieves an approximation ratio of \(\sqrt2-1\approx0.4142\) under either constraint, breaking the previous \(1/3\) barrier. Once an element is selected, its label is chosen from the two labels with the largest marginal gains, with probabilities proportional to their positive parts.

        We further improve the approximation guarantees to
$1/2$ under cardinality and knapsack constraints,
through a Bellman-style
backward induction analysis using the same
proportional Top-2 rule.
Our $1/2$ guarantees for cardinality and
knapsack constraints are asymptotically tight
as $k\to\infty$. They are also tight for $k=2$. The value-oracle complexities are \(O(Bnk)\) for cardinality, where \(B\) is the cardinality bound, \(O(n^2k)\) for matroid, and \(O(n^3k^2)\) for knapsack.

        \par\medskip
        \textbf{Keywords:}
        $k$-submodular maximization;
        Approximation algorithms; Combinatorial algorithms.
    \end{abstract}

\end{titlepage}


\hypersetup{pageanchor=true}
\setcounter{page}{1}
		
		\clearpage
		\hypersetup{pageanchor=true}
		\pagenumbering{arabic}
		\setcounter{page}{1}
		
		\tableofcontents
		\clearpage

		
		\section{Introduction}
		Let \(V= \{e_1,\dots,e_n\}\) be a ground set of \(n\) elements and $2^V$ denote the power set of $V$. A function
		\(f:2^V\to \mathbb{R}\) is \textit{submodular} if, for all \(X,Y\in 2^V\),
		\[
		f(X)+f(Y)\ge f(X\cup Y)+f(X\cap Y).
		\]
		Submodularity captures the principle of diminishing returns and plays an important role in combinatorial optimization and operations research. Classical examples of submodular functions include cut functions, matroid rank functions, and entropy functions. 
		
		Motivated by a question of Lov\'asz~\cite{lovasz1983submodular} on generalizations of submodularity that preserve useful algorithmic properties, \textit{\(k\)-submodularity} has been studied as a natural extension of submodularity. The term was first used by Huber and Kolmogorov~\cite{huber2012towards}, although closely related notions had appeared earlier~\cite{cohen2006complexity}. Informally, a \(k\)-submodular function assigns each selected element of a ground set to one of \(k\) labels and evaluates the resulting \(k\) pairwise disjoint subsets.

		\begin{definition}\label{D1.1}
			Define
			$
			(k+1)^V :=
			\left\{
			(X_1,\ldots,X_k) :
			X_i\subseteq V (\forall  i\in[k]),
			X_i\cap X_j=\emptyset (\forall i\neq j)
			\right\},
			$where \([k]:=\{1,\ldots,k\}\) is a label set.
			For \(\mathbf{x}=(X_1,\ldots,X_k)\) and \(\mathbf{y}=(Y_1,\ldots,Y_k)\) in \((k+1)^V\), define
			\(\mathbf{x}\sqcap\mathbf{y}:=(X_1\cap Y_1,\ldots,X_k\cap Y_k)\) and
			\(\mathbf{x}\sqcup\mathbf{y}:=(Z_1,\ldots,Z_k)\), where
			$
			Z_i := (X_i\cup Y_i)\setminus
			\bigcup_{j\in[k]\setminus\{i\}}(X_j\cup Y_j).
			$
			A function \(f:(k+1)^V\to\mathbb{R}\) is \(k\)-submodular if, for all \(\mathbf{x},\mathbf{y}\in(k+1)^V\),
			\[
			f(\mathbf{x})+f(\mathbf{y})
			\ge
			f(\mathbf{x}\sqcap\mathbf{y})+f(\mathbf{x}\sqcup\mathbf{y}).
			\]
		\end{definition}
		
		Notably,  this definition reduces to a submodular function when $k=1$ and a bisubmodular function when $k=2$~\cite{mccormick2010strongly,huber2014skew}. Throughout this paper, we assume that  \(f(\mathbf{0})=0\), where \(\mathbf{0}:=(\emptyset,\ldots,\emptyset)\). We write \(\mathbf{x}\preceq\mathbf{y}\) if \(X_i\subseteq Y_i\) for all \(i\in[k]\). The function \(f\) is \textit{monotone} if \(f(\mathbf{x})\le f(\mathbf{y})\) whenever \(\mathbf{x}\preceq\mathbf{y}\).
		
		Extensive research has focused on nonnegative \(k\)-submodular maximization under  support constraints. Given a \(k\)-submodular function \(f:(k+1)^V\to\mathbb{R}_{\ge 0}\), the problem can be formulated as
		\[
		\max \left\{ f(\mathbf{x}) : \mathbf{x}\in(k+1)^V \text{ with } \operatorname{supp}(\mathbf{x})\in \mathcal{P} \right\},
		\]
		where \(\operatorname{supp}(\mathbf{x}):=\bigcup_{i\in[k]}X_i\) denotes the \textit{support}  of \(\mathbf{x}\), and \(\mathcal{P}\subseteq 2^V\) is a family of feasible supports.

		When \(\mathcal{P}=2^V\) (the unconstrained setting), the problem is already NP-hard as it captures cut problems such as Max-Cut and Max-\(k\)-Cut~\cite{ward2016maximizing,papadimitriou1988optimization}. For non-monotone \(f\), Ward and \v{Z}ivn\'y~\cite{ward2014maximizing} gave a deterministic \(\frac{1}{3}\)-approximation algorithm. They also analyzed a randomized greedy algorithm that assigns an element to a label with probability proportional to its marginal gain, obtaining a \(1/(1+\alpha)\)-approximation, where \(\alpha=\max\{1,\sqrt{(k-1)/4}\}\). This result provided a reusable randomized greedy framework. By choosing the probability distribution more carefully, Iwata, Tanigawa, and Yoshida~\cite{iwata2016improved} obtained a \(\frac{1}{2}\)-approximation. Oshima~\cite{oshima2021improved} later improved the ratio to \(\frac{k^2+1}{2k^2+1}\) for \(k\ge 3\). For monotone \(f\), the authors in~\cite{iwata2016improved} proved that any \((\frac{k+1}{2k}+\varepsilon)\)-approximation requires exponentially many value queries, and gave a randomized \(\frac{k}{2k-1}\)-approximation algorithm. 
Recently, Chen~\cite{chen2026tight} improved the value-oracle hardness bound to $\frac{k}{2k-1}$, establishing that this approximation ratio is tight.

		When \(\mathcal{P}\) is a cardinality, single matroid, or single knapsack constraint, several algorithms achieve a \(\frac{1}{2}\) or \((\frac{1}{2}-\epsilon)\)-approximation for monotone \(f\)~\cite{ohsaka2015monotone,niu2023fast,sakaue2017maximizing,wang2025approximation}; the multilinear-extension approach further gives a \((\frac{1}{2}-\varepsilon)\)-approximation for a constant number of knapsack constraints~\cite{zhou2025improved}. In view of the  oracle hardness threshold of \(\frac{k+1}{2k}\) mentioned above, these guarantees are asymptotically tight in \(k\).
		For non-monotone \(f\), however, before the first
version of this work, the best-known guarantees under
cardinality, matroid, and knapsack constraints were
\(1/3\) or \(1/3-\varepsilon\). This holds both for combinatorial algorithms based on greedy or local search~\cite{nguyen2020streaming,niu2023fast,sun2022maximization,wang2025approximation} and for multilinear-extension-based algorithms~\cite{zhou2025improved}. The gap between \(\frac{1}{3}\) and \(\frac{k}{2k-1}\) raises a natural question:
		
		\emph{Is \(\frac{1}{3}\) an inherent barrier for constrained non-monotone \(k\)-submodular maximization, or only a limitation of current algorithms and analyses?}
		
		We show that the \(1/3\) barrier is not inherent. Our main result is a unified framework that achieves a \((\sqrt2-1)\)-approximation for non-monotone \(k\)-submodular maximization under either a matroid constraint or a knapsack constraint. This gives the first improvement over the \(1/3\) barrier for these constrained \(k\)-submodular settings. It shows that the additional label structure in \(k\)-submodularity is not merely a source of difficulty: it can also provide exploitable local structure for a simple randomized greedy framework to achieve a \(0.4142\)-approximation. Subsequent continuous approaches based on
multilinear extensions achieved approximation
ratios of \(0.4456-\varepsilon\) for \(k=2\)
and \(0.4541-\varepsilon\) for \(k\ge3\)
under a knapsack constraint~\cite{Agg26},
and \(0.48\) under a matroid
constraint~\cite{WW26}.
In this  new version, building on our proportional Top-2 randomized
greedy framework, we further develop a Bellman-style
backward induction analysis.
The resulting algorithms achieve $1/2$ under
cardinality and knapsack constraints and
$\sqrt2-1$ under a matroid constraint.
In particular, the knapsack guarantee is improved
without changing the original algorithm.
Table~\ref{tab:combinatorial-guarantees} summarizes both analyses and compares them with prior work.
		
		\begin{table}[t]
\caption{Approximation guarantees for constrained $k$-submodular maximization with $k\ge2$.}
\label{tab:combinatorial-guarantees}
\centering
\normalsize
\renewcommand{\arraystretch}{1.18}
\begin{tabular*}{\linewidth}{@{}l@{\extracolsep{\fill}}ccc@{}}
\toprule
Constraints & \begin{tabular}[c]{@{}c@{}}Previous\\(combinatorial)\end{tabular}
& Continuous$^{\dagger}$ & \begin{tabular}[c]{@{}c@{}}This work\\(combinatorial)\end{tabular}\\
\midrule
Cardinality & $1/3$~\cite{nguyen2020streaming} & $0.48$~\cite{WW26}
& $\boldsymbol{1/2}$\\
\addlinespace
Knapsack & $1/3$~\cite{wang2025approximation}
& \begin{tabular}[c]{@{}l@{}}
$0.4456-\varepsilon$ ($k=2$)~\cite{Agg26}\\
$0.4541-\varepsilon$ ($k\ge3$)~\cite{Agg26}
\end{tabular}
& $\boldsymbol{1/2}$\\
\addlinespace
Matroid & $1/3$~\cite{LSL24} & $0.48$~\cite{WW26}
& $\boldsymbol{\sqrt{2}-1\approx 0.414}$\\
\bottomrule
\end{tabular*}
\par\smallskip
{\raggedright\footnotesize
$^{\dagger}$ The continuous results
\cite{Agg26,WW26} appeared after the first
version of this work (September 14, 2026),
which was the first to break the $1/3$
barrier under cardinality, matroid,
and knapsack constraints.
\par}
\end{table}
		
		\begin{theorem}[Original guarantees, informal]\label{theorem1.2}
			For nonnegative non-monotone \(k\)-submodular maximization with \(k\ge 2\), there is a polynomial-time randomized algorithm with the following guarantees:
			\begin{enumerate}
				\item a \((\sqrt{2}-1)\)-approximation under a single matroid constraint, using \(O(n^2k)\) value-oracle queries,
				\item a \((\sqrt{2}-1)\)-approximation under a single knapsack constraint, using \(O(n^3k^2)\) value-oracle queries.
			\end{enumerate}
		\end{theorem}
		
		\begin{theorem}[Bellman refinement, informal]\label{thm:bellman-refinement}
			For every $k\ge2$, maximizing a nonnegative $k$-submodular function $f$ admits randomized combinatorial algorithms with
the following guarantees:
\begin{itemize}
\item \textbf{Cardinality constraints.} A $1/2$-approximation using
$O(Bnk)$ value queries, where $B$ is the cardinality bound.
\item \textbf{Knapsack constraints.} A $1/2$-approximation using
$O(n^3k^2)$ value queries for a single knapsack with 
label-independent costs.
\end{itemize}
		\end{theorem}
		
		The refined guarantees are proved in Theorems~\ref{bellman:thm:cardinality-main} and~\ref{bellman:thm:knapsack-main}. The cardinality and knapsack ratios are oracle-optimal for \(k=2\) and asymptotically tight as \(k\to\infty\).

		\subsection{Technical overview}

         Previous greedy-type algorithms for $k$-submodular
maximization \cite{ohsaka2015monotone,sakaue2017maximizing,wang2025approximation} often establish approximation guarantees
through ``loss-to-gain'' accounting.
Specifically, let $\mathbf{s}_t$ be the partial greedy
solution at step $t$. Such analyses typically
construct an auxiliary hybrid solution $\mathbf{h}_t\succeq\mathbf s_t$
that combines the assignments in $\mathbf{s}_t$
with suitably retained assignments from an
optimal solution $\mathbf o=\mathbf h_0$.
The approximation guarantee is then derived by
relating the loss in hybrid value
$f(\mathbf{h}_t)-f(\mathbf{h}_{t+1})$
to the algorithmic gain
$f(\mathbf{s}_{t+1})-f(\mathbf{s}_t)$.
For monotone functions, adding the selected element
cannot decrease the hybrid value, so only the
deletion loss needs to be controlled, typically
yielding $f(\mathbf h_t)-f(\mathbf h_{t+1})
\le f(\mathbf s_{t+1})-f(\mathbf s_t)$
and hence a $1/2$-approximation.
For non-monotone functions, however, adding the
selected element may also decrease the hybrid value.
By pairwise monotonicity and orthant submodularity,
this additional loss can also be bounded by the
current greedy gain, leading to
$f(\mathbf h_t)-f(\mathbf h_{t+1})
\le 2\bigl(f(\mathbf s_{t+1})-f(\mathbf s_t)\bigr)$,
which limits this analysis to a
$1/3$-approximation guarantee.

To overcome this limitation, we refine both the
label-assignment strategy and the approximation analysis.
On the algorithmic side,
we adopt the \emph{Top-2 rule} for label assignment:
after selecting an element, we randomize over the
two labels with largest marginal gains.
On the analytical side, we decompose the hybrid loss into three components in
Definition~\ref{def:local-losses}:
(i) replacing a retained
optimal label, (ii) adding an element absent from the
hybrid, and (iii) adding such an element while removing
a retained optimal element.
		
		  \vspace{1em}
\noindent\textbf{Local analysis.}
Fix a partial solution \(\mathbf{s}\) and a candidate element \(e\notin \operatorname{supp}(\mathbf{s})\). Let \(\Delta_{e,i}f(\mathbf{s})\) denote the marginal gain of adding \(e\) to \(\mathbf{s}\) with label \(i\in[k]\). By the pairwise monotonicity property recalled in Section~\ref{sec2}, \(\Delta_{e,i}f(\mathbf{s})+\Delta_{e,j}f(\mathbf{s})\ge 0\) for all distinct \(i,j\in[k]\), and hence the loss associated with any residual label can be controlled by the two largest marginal gains. We therefore use only the two best labels; after relabeling, let them be \(1\) and \(2\).
		Define \(y_i:=\max\{\Delta_{e,i}f(\mathbf{s}),0\}\) for \(i=1,2\). The Top-2 rule chooses label \(1\) with probability \(p\) and label \(2\) with probability \(1-p\), where \(0\le p\le 1\). We assign to \(e\) the score \(G_e(\mathbf{s}):=p y_1+(1-p)y_2\). When this score equals the expected marginal gain of the algorithmic step, we call it the exact score.
		
		Suppose that our algorithm chooses \(p\) so that the score of \(e\) is exact. If there exist nonnegative constants \(C,A,Q\) such that the replacement loss is bounded by \(C G_e(\mathbf{s})\), the pure-addition loss by \(A G_e(\mathbf{s})\), and the residual-deletion loss by \(Q G_e(\mathbf{s})\), then the exact-score top-2 randomized greedy framework yields, for both constraint families studied in this work, a stepwise accounting of the form
		\(\operatorname{Loss}\le \eta\,\operatorname{Gain}\), where \(\eta:=\max\{C,A+Q\}\). Indeed, each step can be viewed as an interpolation between a replacement part, charged with coefficient \(C\), and an addition-plus-deletion part, charged with coefficient \(A+Q\). Taking the worst convex combination gives the maximum of these two coefficients.  A standard potential argument then gives a \(1/(1+\eta)\)-approximation.
		
		It remains to choose \(p\) so as to minimize this coefficient $\eta$. We choose the proportional Top-2 rule: if \(y_1=0\), set \(p=1\); otherwise set \(p=y_1/(y_1+y_2)\). Equivalently, the rule takes the two labels with the largest marginal gains, truncates non-positive gains to zero, and chooses among the remaining labels with probability proportional to the truncated gains. Under this proportional Top-2 rule, the score of $e$ is exact, and one can take \(C=1\), \(A=(\sqrt{2}-1)/2\), and \(Q=(\sqrt{2}+1)/2\). Hence \(\eta=\max\{C,A+Q\}=\sqrt{2}\), giving the approximation ratio \(1/(1+\sqrt{2})=\sqrt{2}-1\). We also prove that, under the above form of scalar score, no Top-2 rule, exact-score or not, can achieve \(\eta<\sqrt{2}\).

For the matroid constraint, our original algorithm
selects a maximum-score base in the contracted matroid
and chooses an element uniformly at random from this base.
For the knapsack constraint, we use a one-guess
score-density greedy algorithm.
Combining the proportional Top-2 rule with matroid
base exchange and cost-based residual deletion,
respectively, gives a $(\sqrt2-1)$-approximation
under either constraint.

\vspace{1em}
\noindent\textbf{Bellman-style refinement.}
We continue to use the proportional Top-2 rule.
Rather than requiring the current greedy gain
to cover the hybrid loss, we use Bellman-style
backward induction~\cite{bellman1957dynamic}
to incorporate future gains
into the analysis.
Specifically, let
$D_t=f(\mathbf h_t)-f(\mathbf s_t)$
denote the gap between the hybrid and greedy
solutions.
We choose a nonnegative weight $\lambda_t$
based on the remaining resources and establish
the following potential guarantee:
\[
\mathbb E_t[f(\mathbf s_N)]
\ge f(\mathbf s_t)+\lambda_t[D_t]_+,
\]
where $\mathbb E_t$ denotes conditional expectation
given the history at step $t$, and $N$ is the
final step.
The initial weight
$\lambda_0$ then determines the approximation ratio:
$\mathbb E_0[f(\mathbf s_N)]
\ge f(\mathbf s_0)+\lambda_0[D_0]_+
\ge (1-\lambda_0)f(\mathbf s_0)
+\lambda_0f(\mathbf h_0)
\ge\lambda_0f(\mathbf o)$,
when $0\le \lambda_0\le1$ and $\mathbf h_0=\mathbf o$.

We now apply this potential guarantee to the
cardinality and knapsack constraints.
For each constraint, we specify the greedy
selection rule and the hybrid update, and choose
weights that satisfy the conditions
of the backward-induction analysis.
For the cardinality, we repeatedly select
an unassigned element with maximum expected Top-2 gain.
Since every retained optimal element remains feasible for addition,
no additional deletion is needed, and the hybrid updates
fall into cases (i) and (ii). The potential guarantee with
$\lambda_0=1/2$ then yields a $1/2$-approximation.

For the knapsack constraint, we continue to use our previously
introduced one-guess randomized greedy algorithm, which first
guesses a labeled singleton solution as a seed and then extends
it greedily according to maximum expected Top-2 gain density.
The main difficulty is that some optimal elements
may become infeasible as the budget is consumed.
To address this, we consider the branch seeded by
the highest-value singleton in optimum and remove
a highest-cost element from the remaining optimum.
This ensures that the retained optimal elements
remain feasible during the relevant greedy steps,
while the loss caused by the removal is bounded
by the seed value.
Combining this loss bound with the potential guarantee
using cost-based weights yields
a $1/2$-approximation.

		\subsection{Other related works and organization}
		\noindent \textbf{Non-monotone submodular maximization.}
		The case \(k=1\) has a rich line of research with numerous results. The unconstrained problem admits the tight \(1/2\)-approximation via double greedy~\cite{buchbinder2015tight}. Under standard constraints such as cardinality, matroid, and a constant number of knapsack constraints, the best known general randomized guarantee is \(0.401\)~\cite{buchbinder2024constrained}, while recent deterministic guarantees are \(0.385-\varepsilon\) for matroids and \(0.367-\varepsilon\) for knapsacks~\cite{chen2026deterministic}. On the hardness side, Oveis Gharan and Vondr\'ak~\cite{oveisgharan2011submodular} established oracle inapproximability at \(0.478\) for a matroid constraint and at \(0.491\) for a cardinality constraint; Qi~\cite{qi2022maximizing} subsequently improved the cardinality bound to \(0.478\).
		
		\noindent \textbf{Other constraints for \(k\)-submodular maximization.}
		There is also a line of work on \(k\)-submodular maximization under more general constraint families, such as multiple matroids, matroid-knapsack intersections, and \(p\)-systems~\cite{yu2023maximizing,zhou2025improved,zhang2026maximizing,ye2026ksubmodular}. The known guarantees are typically constraint-dependent; for example, for the intersection of a \(p\)-system and \(d\) knapsack constraints, ratios of the form \((1-\varepsilon)/(p+\alpha+2d)\) and later \((1-\varepsilon)/(p+\alpha+\frac{7}{4}d)\) are known, where \(\alpha=2\) for monotone functions and \(\alpha=3\) for non-monotone functions~\cite{zhang2026maximizing}. 
		
		\noindent \textbf{Online and streaming \(k\)-submodular maximization.}
		Online and streaming variants of \(k\)-submodular maximization have also
		been studied. Soma~\cite{soma2019noregret} gave no-regret algorithms for
		online \(k\)-submodular maximization. Ene and Nguyen~\cite{ene2022streaming}
		developed streaming algorithms for monotone \(k\)-submodular maximization
		under cardinality constraints. More recently, Spaeh, Ene, and
		Nguyen~\cite{spaeh2025online} studied single-pass streaming and online
		algorithms for  \(k\)-submodular maximization under cardinality
		and knapsack constraints.
		
		\vspace{1em}
		
		\noindent \textbf{Organization.}
		This work is organized as follows. Section~\ref{sec2} introduces the notation, basic properties of
		\(k\)-submodular functions, the exact scores and hybrid loss certificates used
		throughout the paper. Section~\ref{sec3} develops the local analysis of Top-2 rules and
		proves that the proportional Top-2 rule attains the \(\sqrt2\)-loss certificate.
		Section~\ref{sec4} applies this certificate to a randomized greedy algorithm under a
		single matroid constraint. Section~\ref{sec5} gives the one-guess score-density greedy
		algorithm and the residual-deletion analysis for a single knapsack constraint.
		Section~\ref{bellman:sec:framework} develops the Bellman refinement and improves the cardinality and knapsack guarantees. Section~\ref{sec6} concludes with several open directions.

		
		\section{Preliminaries}\label{sec2}
		Throughout this paper, we assume that there exists a value oracle \(\mathcal{O}_f\) that returns \(f(\mathbf{x})\) for any query \(\mathbf{x}\in(k+1)^V\). For convenience, we identify \(\mathbf{x}\) with a vector \(\mathbf{x}\in\{0,1,\ldots,k\}^n\), where \(\mathbf{x}_t=j\in[k]\) if \(e_t\in X_j\), and \(\mathbf{x}_t=0\) if \(e_t\notin\operatorname{supp}(\mathbf{x})\). The two notations are used interchangeably when no confusion arises. For \(e\in\operatorname{supp}(\mathbf{x})\), let \(\mathbf{x}_{-e}\) denote deleting \(e\) from \(\operatorname{supp}(\mathbf{x})\), i.e., setting \(\mathbf{x}_e=0\).
		
		We next recall two characteristic properties of \(k\)-submodular functions that will be used throughout the analysis: orthant submodularity and pairwise monotonicity. For \(\mathbf{x}=(X_1,\ldots,X_k)\in(k+1)^V\), \(e\notin\operatorname{supp}(\mathbf{x})\), and \(i\in[k]\), define the marginal gain
		\[
		\Delta_{e,i}f(\mathbf{x})
		:=
		f(X_1,\ldots,X_i\cup\{e\},\ldots,X_k)
		-
		f(X_1,\ldots,X_i,\ldots,X_k).
		\]
		Under the vector representation, this can be written equivalently as
		$
		\Delta_{e,i}f(\mathbf{x})
		=
		f(\mathbf{x}+i\mathbf{1}_e)-f(\mathbf{x}),
		$
		where \(\mathbf{1}_e\in\{0,1\}^V\) denotes the unit vector with value \(1\) at coordinate \(e\) and \(0\) elsewhere.
		
		The function \(f\) is \textit{orthant submodular} if
		\[
		\Delta_{e,i}f(\mathbf{x})\ge \Delta_{e,i}f(\mathbf{y})
		\quad
		\text{for all } \mathbf{x}\preceq\mathbf{y},\ e\notin\operatorname{supp}(\mathbf{y}),\ \text{and } i\in[k],
		\]
		and it is \textit{pairwise monotone} if
		\[
		\Delta_{e,i}f(\mathbf{x})+\Delta_{e,j}f(\mathbf{x})\ge 0
		\quad
		\text{for all } \mathbf{x}\in(k+1)^V,\ e\notin\operatorname{supp}(\mathbf{x}),\ \text{and distinct } i,j\in[k].
		\]
		
		Ward and \v{Z}ivn\'y~\cite{ward2016maximizing} showed that these two properties indeed characterize \(k\)-submodular functions.
		
		\begin{theorem}[Ward and \v{Z}ivn\'y~\cite{ward2016maximizing}]
			Let \(k\ge 2\). A function \(f:(k+1)^V\to \mathbb{R}\) is \(k\)-submodular if and only if it is orthant submodular and pairwise monotone.
		\end{theorem}
		
		To make the two constraint families considered in this paper explicit, we introduce the following notation and terminology.
		\vspace{1em}
		
		\noindent \textbf{Matroids and contractions.}
		A matroid \(\mathcal{M}=(V,\mathcal{P})\) is a set system, where
		\(\mathcal{P}\subseteq 2^V\) is a family of independent sets satisfying:
		(i) \(\emptyset\in\mathcal{P}\);
		(ii) if \(A\subseteq B\) and \(B\in\mathcal{P}\), then \(A\in\mathcal{P}\);
		(iii) if \(A,B\in\mathcal{P}\) with \(|A|<|B|\), then there exists
		\(e\in B\setminus A\) such that \(A\cup\{e\}\in\mathcal{P}\).
		A maximal independent set is called a \textit{base}, and all bases have the same cardinality,
		called the rank of \(\mathcal{M}\). A solution \(\mathbf{x}\in(k+1)^V\) is feasible
		under \(\mathcal{M}\) if \(\operatorname{supp}(\mathbf{x})\in\mathcal{P}\).
		
		For an independent set \(S\in\mathcal{P}\), the contraction \(\mathcal{M}/S\) is the
		matroid on \(V\setminus S\) whose independent sets are
		$
		\{T\subseteq V\setminus S: S\cup T\in\mathcal{P}\}.
		$
		Equivalently, the bases of \(\mathcal{M}/S\) are exactly the sets
		\(M\subseteq V\setminus S\) such that \(S\cup M\) is a base of \(\mathcal{M}\).
		
		\vspace{1em}
		
		\noindent \textbf{Knapsack constraints.}
		Each element \(e\in V\) has a nonnegative cost \(c_e\). The feasible support family under a single knapsack constraint is
		$
		\mathcal{P}
		:=
		\left\{
		S\subseteq V:
		\sum_{e\in S}c_e\le 1
		\right\},
		$
		where the knapsack budget is normalized to one. Since elements with \(c_e>1\) are infeasible, we may assume w.l.o.g. that
		\(c_e\in[0,1]\) for all \(e\in V\). For a solution
		\(\mathbf{x}\in(k+1)^V\), we write
		$
		c(\mathbf{x})
		:=
		\sum_{e\in\operatorname{supp}(\mathbf{x})} c_e 
		$
		for simplicity.
		
		\subsection{Exact scores and hybrid loss certificates}
		
		We first formalize the local randomized rules used by our algorithms.
		
		\begin{definition}[Exact-score local rule]\label{def:exact-score-rule}
			Given a current algorithmic solution \(\mathbf{s}\in(k+1)^V\) and an unassigned element \(e\notin\operatorname{supp}(\mathbf{s})\), a local randomized rule specifies a distribution \(P_e(\mathbf{s})\) over labels in \([k]\). Its conditional expected marginal gain is
			$
			g_e(\mathbf{s})
			:=
			\mathbb{E}_{i\sim P_e(\mathbf{s})}
			\bigl[\Delta_{e,i}f(\mathbf{s})\bigr].
			$
			The rule is called an \textit{exact-score rule} if the algorithm assigns to \(e\) a nonnegative score \(G_e(\mathbf{s})\) satisfying
			$
			G_e(\mathbf{s})=g_e(\mathbf{s})
			$
			for every state \(\mathbf{s}\) and every \(e\notin\operatorname{supp}(\mathbf{s})\).
		\end{definition}
		
		We next define the hybrid solutions used in the analysis.
		
		\begin{definition}[Hybrid solution]\label{def:hybrid-solution}
			Let \(\mathbf{s}=(S_1,\ldots,S_k)\in(k+1)^V\) be the current algorithmic solution and let \(\mathbf{o}=(O_1,\ldots,O_k)\in(k+1)^V\) be the current residual comparison solution maintained by the analysis, often chosen as a part of an optimal solution. The hybrid solution \(\mathbf{h}=\mathbf{h}(\mathbf{s},\mathbf{o})\) is defined by
			\[
			\mathbf{h}(\mathbf{s},\mathbf{o})
			:=
			(H_1,\ldots,H_k),
			\qquad
			H_i:=S_i\cup\bigl(O_i\setminus\operatorname{supp}(\mathbf{s})\bigr)
			\quad\text{for }i\in[k].
			\]
			Equivalently, in vector notation, \(\mathbf{h}_e=\mathbf{s}_e\) if \(\mathbf{s}_e\neq0\), and \(\mathbf{h}_e=\mathbf{o}_e\) otherwise.
		\end{definition}
		
		Thus the hybrid solution keeps all current algorithmic assignments and the residual comparison  assignments that have not yet been overwritten.
		
		We now isolate the three local losses appearing in the hybrid analysis. 
		
		\begin{definition}[Local hybrid losses]\label{def:local-losses}
			Fix one step of the algorithm. Let \(\mathbf{s}\) be the current algorithmic solution, \(\mathbf{o}\) the current residual comparison solution, and \(\mathbf{h}=\mathbf{h}(\mathbf{s},\mathbf{o})\) the corresponding hybrid solution. Suppose that the algorithm selects an element \(e\notin\operatorname{supp}(\mathbf{s})\) and assigns it a random label \(i\sim P_e(\mathbf{s})\).
			
			If \(e\in\operatorname{supp}(\mathbf{o})\), let \(\ell:=\mathbf{o}_e\). Since \(\mathbf{h}=\mathbf{h}_{-e}+\ell\mathbf{1}_e\), replacing the residual label \(\ell\) by the randomized algorithmic label incurs the replacement loss
			\[
			L_{\mathrm{rep}}
			:=
			\Delta_{e,\ell}f(\mathbf{h}_{-e})
			-
			\mathbb{E}_{i\sim P_e(\mathbf{s})}
			\bigl[\Delta_{e,i}f(\mathbf{h}_{-e})\bigr].
			\]
			If \(e\notin\operatorname{supp}(\mathbf{o})\), adding \(e\) to the hybrid solution incurs the pure-addition loss
			\[
			L_{\mathrm{add}}
			:=
			f(\mathbf{h})
			-
			\mathbb{E}_{i\sim P_e(\mathbf{s})}
			\bigl[f(\mathbf{h}+i\mathbf{1}_e)\bigr]
			=
			-
			\mathbb{E}_{i\sim P_e(\mathbf{s})}
			\bigl[\Delta_{e,i}f(\mathbf{h})\bigr].
			\]
			Still in the case \(e\notin\operatorname{supp}(\mathbf{o})\), maintaining feasibility may require deleting a residual element \(r\in\operatorname{supp}(\mathbf{o})\setminus\operatorname{supp}(\mathbf{s})\). Let \(\ell:=\mathbf{o}_r\). Conditional on assigning \(e\) label \(i\), the residual-deletion loss is
			\[
			L_{\mathrm{del}}
			:=
			f(\mathbf{h}+i\mathbf{1}_e)
			-
			f(\mathbf{h}_{-r}+i\mathbf{1}_e)
			=
			\Delta_{r,\ell}f(\mathbf{h}_{-r}+i\mathbf{1}_e).
			\]
		\end{definition}
		
		Since \(\mathbf{s}\preceq \mathbf{h}_{-r}+i\mathbf{1}_e\), orthant submodularity gives
		\[
		L_{\mathrm{del}}
		\le
		\Delta_{r,\ell}f(\mathbf{s}).
		\]
		Thus the residual-deletion loss is controlled by the current marginal gain of the deleted element.
		
		The following definition abstracts the above estimates into pointwise local bounds. The elements appearing in the replacement and pure-addition bounds are possible algorithmic choices, whereas the element in the deletion bound is a possible residual element deleted by an exchange argument; these elements need not be the same.
		
		\begin{definition}[\((C,A,Q)\)-certificate]\label{d2.5}
			Let \(P\) be a local randomized rule, and let \(G\) be its associated nonnegative score. We say that \(P\) admits a \((C,A,Q)\)-certificate if the following three inequalities hold for every \(k\)-submodular function \(f\).
			\begin{enumerate}
				\item \textit{\(C\)-replacement bound:} for all \(\mathbf{s}\preceq\mathbf{h}\), all \(e\notin\operatorname{supp}(\mathbf{h})\), and all \(\ell\in[k]\),
				\[
				\Delta_{e,\ell}f(\mathbf{h})
				-
				\mathbb{E}_{i\sim P_e(\mathbf{s})}
				\bigl[\Delta_{e,i}f(\mathbf{h})\bigr]
				\le
				C\,G_e(\mathbf{s}).
				\]
				
				\item \textit{\(A\)-addition bound:} for all \(\mathbf{s}\preceq\mathbf{h}\) and all \(e\notin\operatorname{supp}(\mathbf{h})\),
				\[
				-
				\mathbb{E}_{i\sim P_e(\mathbf{s})}
				\bigl[\Delta_{e,i}f(\mathbf{h})\bigr]
				\le
				A\,G_e(\mathbf{s}).
				\]
				
				\item \textit{\(Q\)-deletion domination:} for all \(\mathbf{s}\in(k+1)^V\), all \(r\notin\operatorname{supp}(\mathbf{s})\), and all \(\ell\in[k]\),
				\[
				\Delta_{r,\ell}f(\mathbf{s})
				\le
				Q\,G_r(\mathbf{s}).
				\]
			\end{enumerate}
		\end{definition}

		\section{Local Analysis of Top-2 Rules}\label{sec3}
		This section analyzes the local structure associated with a single element. For a general Top-2 rule that only uses the two labels with the largest marginal gains, we derive upper bounds for the replacement and pure-addition losses and show that these bounds are tight on suitable \(k\)-submodular instances. We then prove that, under the scalar score used in this section, any constants \(C,A,Q\) satisfying the certificate in Definition~\ref{d2.5} must have \(A+Q\ge \sqrt{2}\). Finally, we show that the simple proportional Top-2 rule is an exact-score rule and achieves
		$
		C=1,
		A=\frac{\sqrt{2}-1}{2},
		Q=\frac{\sqrt{2}+1}{2},
		$
		so that \(A+Q=\sqrt{2}\) is attainable.
		
		Fix a current state \(\mathbf{s}\in(k+1)^V\) and an element \(e\) with \(\mathbf{s}_e=0\). Let
		$
		w_i=\Delta_{e,i}f(\mathbf{s})
		\text{ for all } i\in[k].
		$
		Choose the two labels with the largest marginal gains and relabel them as \(1\) and \(2\), so that \(w_1\ge w_2\ge w_j\) for all \(j\notin\{1,2\}\). Define
		$
		y_1=\max\{w_1,0\}
		\text{ and }
		y_2=\max\{w_2,0\}.
		$
		
		\begin{lemma}[Zero-score case]\label{lem:zero-score}
			If \(y_1=0\), then \(w_i=0\) for all \(i\in[k]\). Moreover, if \(\mathbf{s}\preceq\mathbf{h}\) and \(\mathbf{h}_e=0\), then \(\Delta_{e,i}f(\mathbf{h})=0\) for all \(i\in[k]\).
		\end{lemma}
		
		\begin{proof}
			Since \(y_1=0\) and \(w_1\) is the largest marginal gain, we have \(w_i\le 0\) for all \(i\in[k]\). If \(w_j<0\) for some \(j\), then for any \(i\neq j\), we also have \(w_i\le 0\), and hence \(w_i+w_j<0\), contradicting pairwise monotonicity. Thus \(w_i=0\) for all \(i\in[k]\).
			
			Now suppose that \(\mathbf{s}\preceq\mathbf{h}\) and \(h_e=0\). By orthant submodularity, it follows that
			$
			\Delta_{e,j}f(\mathbf{h})\le \Delta_{e,j}f(\mathbf{s})=0.
			$
			If \(\Delta_{e,j}f(\mathbf{h})<0\) for some \(j\), then for any \(i\neq j\), we have \(\Delta_{e,i}f(\mathbf{h})\le 0\), so
			$
			\Delta_{e,i}f(\mathbf{h})+\Delta_{e,j}f(\mathbf{h})<0,
			$
			again contradicting pairwise monotonicity.
		\end{proof}
		
		In the remainder of this section, assume that \(y_1>0\). A Top-2 rule, not necessarily exact-score, randomizes only over labels \(1\) and \(2\). Let the probability distribution be
		\[
		\Pr[\mathbf{s}_e=1]=p,\qquad
		\Pr[\mathbf{s}_e=2]=1-p,\qquad
		0\le p\le 1.
		\]
		We assign to element \(e\) the score
		$
		G_e(\mathbf{s})=p y_1+(1-p)y_2.
		$
		By homogeneity, normalize
		\[
		y_1=1,\qquad
		y_2=z,\qquad
		0\le z\le 1.
		\]
		Then the normalized score of \(e\) is \(G(z,p)=p+(1-p)z\). Clearly, \(G_e(\mathbf{s})=y_1G(z,p)\).
		
		For any \(\mathbf{s}\preceq\mathbf{h}\) with \(h_e=0\), write
		$
		a_i=\Delta_{e,i}f(\mathbf{h}),
		\text{ for all } i\in[k].
		$
		By orthant submodularity and the definition of the top-2 labels, we have
		\begin{equation*}
			a_1\le 1,
			\qquad
			a_2\le z,
			\qquad
			a_\ell\le z \quad (\ell\notin\{1,2\}),
		\end{equation*}
		and pairwise monotonicity gives
		\begin{equation}
			a_i+a_j\ge 0,
			\qquad
			i,j\in[k],\ i\neq j.
			\label{pair_m}
		\end{equation}
		It is immediate that
		$
		\Delta_{e,\ell}f(\mathbf{h})
		-
		\mathbb{E}_{i\sim P_e(\mathbf{s})}
		\bigl[\Delta_{e,i}f(\mathbf{h})\bigr]
		=
		a_\ell-\bigl(pa_1+(1-p)a_2\bigr)
		$
		and
		$
		-
		\mathbb{E}_{i\sim P_e(\mathbf{s})}
		\bigl[\Delta_{e,i}f(\mathbf{h})\bigr]
		=
		-\bigl(pa_1+(1-p)a_2\bigr).
		$
		
		\subsection{Replacement/addition bounds and \texorpdfstring{\(Q\)}{Q}-certificate}
		Under the above normalized notation, the replacement and pure-addition losses admit the following local upper bounds. In addition, the \(Q\)-deletion domination certificate is characterized by the following necessary and sufficient condition.
		
		\begin{lemma}\label{lem:local-upper-envelopes} 
			For any \(p\in[0,1]\), the following statements hold.
			\begin{enumerate}
				\item[(i)] For any residual label \(\ell\in[k]\), the replacement loss satisfies
				\begin{equation}
					a_\ell-\bigl(pa_1+(1-p)a_2\bigr)
					\le
					\max\{2(1-p),2pz\}.
					\label{rep}
				\end{equation}
				
				\item[(ii)] The pure-addition loss satisfies
				\begin{equation}
					-\bigl(pa_1+(1-p)a_2\bigr)
					\le
					U(z,p),
					\label{add}
				\end{equation}
				where
				\begin{equation}
					U(z,p)=
					\begin{cases}
						1-2p, & 0\le p\le \frac{1}{2},\\
						(2p-1)z, & \frac{1}{2}\le p\le 1.
					\end{cases}
					\label{U}
				\end{equation}
				
				\item[(iii)] For fixed \(z\) and \(p\), the score-domination condition for current marginal gains,
				\[
				\Delta_{e,\ell}f(\mathbf{s})\le QG_e(\mathbf{s}),
				\qquad  \ell\in[k],
				\]
				holds for all \(k\)-submodular local instances if and only if
				\begin{equation}
					Q\ge \frac{1}{G(z,p)}.
					\label{Q}
				\end{equation}
				For a local randomized rule \(p=p(z)\) depending on \(z\), a uniform \(Q\)-certificate holds if and only if
				\begin{equation}
					Q\ge \sup_{0\le z\le 1}\frac{1}{G(z,p(z))}.
					\label{Q1}
				\end{equation}
			\end{enumerate}
		\end{lemma}
		
		\begin{proof}
			We first prove \emph{(i)}. If \(\ell=1\), then
			$
			a_1-\bigl(pa_1+(1-p)a_2\bigr)
			=
			(1-p)(a_1-a_2).
			$
			By~\eqref{pair_m}, we have \(a_2\ge -a_1\), and hence
			\(a_1-a_2\le 2a_1\le 2\), where the last inequality uses \(a_1\le 1\). Thus, in this case, the left-hand side of~\eqref{rep} is at most \(2(1-p)\).
			
			If \(\ell=2\), then
			$
			a_2-\bigl(pa_1+(1-p)a_2\bigr)
			=
			p(a_2-a_1).
			$
			By~\eqref{pair_m}, we have \(a_1\ge -a_2\), and hence
			\(a_2-a_1\le 2a_2\le 2z\). Thus, in this case, the left-hand side of~\eqref{rep} is at most \(2pz\).
			
			It remains to consider \(\ell\notin\{1,2\}\). Let \(a:=a_\ell\). If \(a<0\), then~\eqref{pair_m} gives \(a_1\ge -a\) and \(a_2\ge -a\), so
			$
			pa_1+(1-p)a_2\ge -a.
			$
			Therefore,
			$
			a-\bigl(pa_1+(1-p)a_2\bigr)\le 2a<0.
			$
			Since the right-hand side of~\eqref{rep} is nonnegative, the claim follows in this case.
			
			Now suppose that \(a\ge 0\). If \(p\ge \frac{1}{2}\), then, under the constraints
			\(a_1\ge -a\), \(a_2\ge -a\), and \(a_1+a_2\ge 0\), the linear function
			\(pa_1+(1-p)a_2\) is minimized at \((a_1,a_2)=(-a,a)\), with value
			\((1-2p)a\). Hence
			\[
			a-\bigl(pa_1+(1-p)a_2\bigr)
			\le
			a-(1-2p)a
			=
			2pa
			\le
			2pz,
			\]
			where we use \(a=a_\ell\le z\). If \(p\le \frac{1}{2}\), then the same linear function is minimized at \((a_1,a_2)=(a,-a)\), with value \((2p-1)a\). Hence
			\[
			a-\bigl(pa_1+(1-p)a_2\bigr)
			\le
			a-(2p-1)a
			=
			2(1-p)a
			\le
			2(1-p)z
			\le
			2(1-p).
			\]
			Combining the three cases proves \emph{(i)}.
			
			We next prove \emph{(ii)}. If \(p\ge \frac{1}{2}\), then \(a_1+a_2\ge 0\) implies \(a_1\ge -a_2\). Hence
			\[
			pa_1+(1-p)a_2
			\ge
			p(-a_2)+(1-p)a_2
			=
			(1-2p)a_2.
			\]
			Since \(1-2p\le 0\) and \(a_2\le z\), we have
			$
			(1-2p)a_2\ge (1-2p)z.
			$
			Therefore,
			$
			pa_1+(1-p)a_2\ge (1-2p)z,
			$
			which gives
			$
			-\bigl(pa_1+(1-p)a_2\bigr)\le (2p-1)z.
			$
			
			If \(p\le \frac{1}{2}\), then \(a_1+a_2\ge 0\) implies \(a_2\ge -a_1\). Thus
			\[
			pa_1+(1-p)a_2
			\ge
			pa_1+(1-p)(-a_1)
			=
			(2p-1)a_1.
			\]
			Since \(2p-1\le 0\) and \(a_1\le 1\), we have
			$
			(2p-1)a_1\ge 2p-1.
			$
			Therefore,
			$
			pa_1+(1-p)a_2\ge 2p-1,
			$
			which gives
			$
			-\bigl(pa_1+(1-p)a_2\bigr)\le 1-2p.
			$
			Then the statement of \emph{(ii)} follows.
			
			We finally prove \emph{(iii)}. For sufficiency, note that \(w_i\le y_1\) for every \(i\in[k]\). If \(Q\ge 1/G(z,p)\), then, since \(G_e(\mathbf{s})=y_1G(z,p)\), we have
			$
			QG_e(\mathbf{s})
			=
			y_1QG(z,p)
			\ge
			y_1
			\ge
			\Delta_{e,\ell}f(\mathbf{s}).
			$
			
			For necessity, consider a $k$-submodular instance on the singleton ground set \(V=\{e\}\). Fix the current state \(\mathbf{s}=\mathbf{0}\), and define
			\[
			f(\mathbf{0})=0,\qquad
			f((1))=1,\qquad
			f((2))=z,
			\]
			and, if \(k>2\), set \(f( (j))=0\) for all \(j\ge3\). It is easy to verify that \(f\) is orthant submodular and pairwise monotone. The marginal gains are \(1\) and \(z\), and hence \(G_e(\mathbf{s})=G(z,p)\). From
			$
			\Delta_{e,1}f(\mathbf{s})=1\le QG_e(\mathbf{s}),
			$
			we obtain~\eqref{Q}. Taking the supremum over \(z\in[0,1]\) gives~\eqref{Q1}.
		\end{proof}
		
		We next use two simple bisubmodular functions on the ground set \(V=\{e_1,e_2\}\) to show that inequalities~\eqref{rep} and~\eqref{add} can be tight. Since bisubmodular functions are the \(k=2\) case of \(k\)-submodular functions, these examples already rule out improving the local bounds in the general \(k\)-submodular setting.
		
		\begin{lemma}\label{lem:sharpness-atom}
			For any \(z\in[0,1]\), each of the following two tables defines a bisubmodular function. The entries in the tables give the function values when \(e_1\) and \(e_2\) take the corresponding labels.
			
			\medskip
			\noindent\textbf{Instance A.}
			\begin{center}
				\renewcommand{\arraystretch}{1.15}
				\setlength{\tabcolsep}{8pt}
				\begin{tabular}{c|ccc}
					\diagbox[width=7em,height=2.8em]{\(e_1\)-label}{\(e_2\)-label}
					& \(0\) & \(1\) & \(2\)\\
					\hline
					\(0\) & \(0\) & \(1+z\) & \(0\)\\
					\(1\) & \(1\) & \(1\) & \(1\)\\
					\(2\) & \(z\) & \(1+2z\) & \(z\)
				\end{tabular}
			\end{center}
			At the state \(\mathbf{s}=(\mathbf{s}_{1}=0,\mathbf{s}_{2}=0)\), the current marginal-gain vector of \(e_1\) is \((w_1,w_2)=(1,z)\). At the extended state \(\mathbf{h}=(\mathbf{h}_{1}=0,\mathbf{h}_{2}=1)\), the marginal-gain vector of \(e_1\) is \((a_1,a_2)=(-z,z)\). Therefore:
			\begin{enumerate}
				\item[(a)] If the label of \(e_1\) in the state \(\hat{\mathbf{h}}=(2,1)\) is replaced, then a direct calculation gives the replacement loss
				\[
				a_2-\bigl(pa_1+(1-p)a_2\bigr)=2pz.
				\]
				
				\item[(b)] If \(e_1\) is added to \(\mathbf{h}\), then for \(p\ge \frac{1}{2}\), a direct calculation gives the pure-addition loss 
				\[
				-\bigl(pa_1+(1-p)a_2\bigr)=(2p-1)z.
				\]
			\end{enumerate}
			
			\noindent\textbf{Instance B.}
			\begin{center}
				\renewcommand{\arraystretch}{1.15}
				\setlength{\tabcolsep}{8pt}
				\begin{tabular}{c|ccc}
					\diagbox[width=7em,height=2.8em]{\(e_1\)-label}{\(e_2\)-label}
					& \(0\) & \(1\) & \(2\)\\
					\hline
					\(0\) & \(0\) & \(1+z\) & \(0\)\\
					\(1\) & \(1\) & \(2+z\) & \(1\)\\
					\(2\) & \(z\) & \(z\) & \(z\)
				\end{tabular}
			\end{center}
			At the state \(\mathbf{s}=(\mathbf{s}_{1}=0,\mathbf{s}_{2}=0)\), the current marginal-gain vector of \(e_1\) is \((w_1,w_2)=(1,z)\). At the extended state \(\mathbf{h}=(\mathbf{h}_{1}=0,\mathbf{h}_{2}=1)\), the marginal-gain vector of \(e_1\) is \((a_1,a_2)=(1,-1)\). Therefore:
			\begin{enumerate}
				\item[(a)] If the label of \(e_1\) in the state \(\overline{\mathbf{h}}=(1,1)\) is replaced, then a direct calculation gives the replacement loss 
				\[
				a_1-\bigl(pa_1+(1-p)a_2\bigr)=2(1-p).
				\]
				
				\item[(b)] If \(e_1\) is added to \(\mathbf{h}\), then for \(p\le \frac{1}{2}\), a direct calculation gives the pure-addition loss 
				\[
				-\bigl(pa_1+(1-p)a_2\bigr)=1-2p.
				\]
			\end{enumerate}
		\end{lemma}
		
		\begin{proof}
			It suffices to verify that the functions defined by Instances A and B are orthant submodular and pairwise monotone.
			
			We first check Instance A. Fixing the label of \(e_2\) or $e_1$, the marginal-gain vectors of \(e_1\) are
			\[
			\mathbf{s}_{2}=0:(1,z),\qquad
			\mathbf{s}_{2}=1:(-z,z),\qquad
			\mathbf{s}_{2}=2:(1,z),
			\]
			and the marginal-gain vectors of \(e_2\) are
			\[
			\mathbf{s}_{1}=0:(1+z,0),\qquad
			\mathbf{s}_{1}=1:(0,0),\qquad
			\mathbf{s}_{1}=2:(1+z,0).
			\]
			At each state, the sum of the two marginal gains of the same element is nonnegative, so pairwise monotonicity holds. For orthant submodularity, the pointwise comparisons are
			\[
			\begin{array}{@{}l@{\qquad}l@{}}
				\Delta_{e_1,1}f(0,0)=1\ge -z=\Delta_{e_1,1}f(0,1),
				&
				\Delta_{e_2,1}f(0,0)=1+z\ge 0=\Delta_{e_2,1}f(1,0),
				\\[1mm]
				\Delta_{e_1,2}f(0,0)=z\ge z=\Delta_{e_1,2}f(0,1),
				&
				\Delta_{e_2,2}f(0,0)=0\ge 0=\Delta_{e_2,2}f(1,0),
				\\[1mm]
				\Delta_{e_1,1}f(0,0)=1\ge 1=\Delta_{e_1,1}f(0,2),
				&
				\Delta_{e_2,1}f(0,0)=1+z\ge 1+z=\Delta_{e_2,1}f(2,0),
				\\[1mm]
				\Delta_{e_1,2}f(0,0)=z\ge z=\Delta_{e_1,2}f(0,2),
				&
				\Delta_{e_2,2}f(0,0)=0\ge 0=\Delta_{e_2,2}f(2,0).
			\end{array}
			\]
			Hence Instance A is orthant submodular and pairwise monotone.
			
			We next check Instance B. Fixing the label of \(e_2\) or $e_1$, the marginal-gain vectors of \(e_1\) are
			\[
			\mathbf{s}_{2}=0:(1,z),\qquad
			\mathbf{s}_{2}=1:(1,-1),\qquad
			\mathbf{s}_{2}=2:(1,z),
			\]
			and the marginal-gain vectors of \(e_2\) are
			\[
			\mathbf{s}_{1}=0:(1+z,0),\qquad
			\mathbf{s}_{1}=1:(1+z,0),\qquad
			\mathbf{s}_{1}=2:(0,0).
			\]
			Again, the sum of the two marginal gains of the same element is nonnegative at every state. The pointwise comparisons for orthant submodularity are
			\[
			\begin{array}{@{}l@{\qquad}l@{}}
				\Delta_{e_1,1}f(0,0)=1\ge 1=\Delta_{e_1,1}f(0,1),
				&
				\Delta_{e_2,1}f(0,0)=1+z\ge 1+z=\Delta_{e_2,1}f(1,0),
				\\[1mm]
				\Delta_{e_1,2}f(0,0)=z\ge -1=\Delta_{e_1,2}f(0,1),
				&
				\Delta_{e_2,2}f(0,0)=0\ge 0=\Delta_{e_2,2}f(1,0),
				\\[1mm]
				\Delta_{e_1,1}f(0,0)=1\ge 1=\Delta_{e_1,1}f(0,2),
				&
				\Delta_{e_2,1}f(0,0)=1+z\ge 0=\Delta_{e_2,1}f(2,0),
				\\[1mm]
				\Delta_{e_1,2}f(0,0)=z\ge z=\Delta_{e_1,2}f(0,2),
				&
				\Delta_{e_2,2}f(0,0)=0\ge 0=\Delta_{e_2,2}f(2,0).
			\end{array}
			\]
			Therefore Instance B is also  bisubmodular.
		\end{proof}
		
		\subsection{\texorpdfstring{$A+Q \ge \sqrt{2}$}{A + Q >= sqrt(2)}}
		For fixed \(z\) and \(p\), Lemma~\ref{lem:local-upper-envelopes}\emph{(ii)} and the tight instances in Lemma~\ref{lem:sharpness-atom} imply that, if the \(A\)-addition bound in Definition~\ref{d2.5} holds, then
		$
		U(z,p)\le A G(z,p),
		$
		or equivalently,
		$
		A\ge \frac{U(z,p)}{G(z,p)}.
		$
		Moreover, if the \(Q\)-deletion domination holds, then Lemma~\ref{lem:local-upper-envelopes}\emph{(iii)} gives
		$
		Q\ge \frac{1}{G(z,p)}.
		$
		Therefore,
		\begin{equation}\label{eq:AQ-lower}
			A+Q
			\ge
			\frac{1+U(z,p)}{G(z,p)}.
		\end{equation}
		The above inequality must hold for all \(z\in[0,1]\). For each fixed \(z\), however, the rule may choose \(p=p(z)\) so as to minimize the right-hand side. Thus the problem reduces to the following max-min expression:
		\[
		\sup_{0\le z\le 1}\inf_{0\le p\le 1}
		\frac{1+U(z,p)}{G(z,p)}.
		\]
		
		\begin{lemma}\label{lem:AQ-lower-bound}
			Consider any Top-2 rule, with the scalar score defined in this section.
			If the \(A\)-addition bound and the \(Q\)-deletion domination certificate in
			Definition~\ref{d2.5} hold uniformly over all normalized $k$-submodular local instances, then
			$
			A+Q\ge \sqrt{2}.
			$
		\end{lemma}
		
		\begin{proof}
			Fix \(z\in[0,1]\). We minimize the right-hand side of \eqref{eq:AQ-lower} over \(p\).
			
			If \(0\le p\le \frac{1}{2}\), then by~\eqref{U}, we have \(U(z,p)=1-2p\). Hence
			$
			\frac{1+U(z,p)}{G(z,p)}
			=
			\frac{2-2p}{z+p(1-z)}.
			$
			Taking the derivative with respect to \(p\), we get
			\[
			\frac{\partial}{\partial p}
			\frac{2-2p}{z+p(1-z)}
			=
			\frac{-2(z+p(1-z))-(2-2p)(1-z)}
			{(z+p(1-z))^2}
			=
			\frac{-2}{(z+p(1-z))^2}<0.
			\]
			Thus the minimum on this interval is attained at \(p=\frac{1}{2}\), with value
			$
			\frac{2}{1+z}.
			$
			
			If \(\frac{1}{2}\le p\le 1\), then \(U(z,p)=(2p-1)z\). Let
			$
			F(p,z)
			=
			\frac{1+(2p-1)z}{p+(1-p)z}.
			$
			Writing \(N=1-z+2pz\) and \(D=z+p(1-z)>0\), we have
			$
			\frac{\partial F}{\partial p}
			=
			\frac{2zD-N(1-z)}{D^2}.
			$
			The numerator expands as
			\[
			2zD-N(1-z)
			=
			2z\bigl(z+p(1-z)\bigr)
			-
			(1-z+2pz)(1-z)
			=
			z^2+2z-1.
			\]
			Let \(z_*=\sqrt{2}-1\), so that \(z_*^2+2z_*-1=0\). If \(0\le z<z_*\), then \(F\) is decreasing in \(p\), and its minimum is attained at \(p=1\), with value
			$
			F(1,z)=1+z.
			$
			Moreover, \(1+z\le \frac{2}{1+z}\) in this range. If \(z_*<z\le 1\), then \(F\) is increasing in \(p\), and its minimum is attained at \(p=\frac{1}{2}\), with value
			$
			F\left(\frac{1}{2},z\right)=\frac{2}{1+z}.
			$
			When \(z=z_*\), the function \(F\) is constant in \(p\), and
			$
			1+z=\frac{2}{1+z}.
			$
			
			Combining the two ranges of \(p\), for each fixed \(z\), no Top-2 rule can reduce the right-hand side of \eqref{eq:AQ-lower} below
			\[
			m(z)=
			\begin{cases}
				1+z, & 0\le z\le \sqrt{2}-1,\\[1mm]
				\dfrac{2}{1+z}, & \sqrt{2}-1\le z\le 1.
			\end{cases}
			\]
			The function \(1+z\) is increasing on \([0,\sqrt{2}-1]\), while \(2/(1+z)\) is decreasing on \([\sqrt{2}-1,1]\), and the two values coincide at \(z=z_*\), where they equal \(\sqrt{2}\). Hence
			$
			\max_{0\le z\le 1}m(z)=\sqrt{2}.
			$
			Therefore,
			$
			A+Q\ge \max_{0\le z\le 1}m(z)=\sqrt{2}.
			$
		\end{proof}
		
		\subsection{Proportional Top-2 rule}\label{sec:proportional-top2}
		We now construct a simple rule attaining the lower bound in Lemma~\ref{lem:AQ-lower-bound}. Set the two terms on the right-hand side of Lemma~\ref{lem:local-upper-envelopes}\emph{(i)} equal, namely \(2(1-p)=2pz\). This linear equation uniquely gives
		\begin{equation}\label{proportional}
			p(z)=\frac{1}{1+z},
			\qquad
			1-p(z)=\frac{z}{1+z}.
		\end{equation}
		Equivalently, the probabilities of the two largest marginal-gain labels are proportional to their current nonnegative marginal gains: if \(y_1+y_2>0\), set
		$
		p_1=\frac{y_1}{y_1+y_2}
		\text{ and }
		p_2=\frac{y_2}{y_1+y_2}.
		$
		If \(y_1=0\), then by Lemma~\ref{lem:zero-score}, all current marginal gains are zero, and the label choice is irrelevant; we set \(p_1=1\) and \(p_2=0\). We call this rule the \emph{proportional Top-2 rule}.
		
		The proportional Top-2 rule is also an exact-score rule. Indeed, if \(y_1=0\), then both the expected marginal gain and the score are zero. If \(y_1>0\), the rule assigns positive probability only to top-2 labels with \(y_i>0\), and hence every selected label satisfies \(y_i=w_i\). Therefore,
		$
		\mathbb{E}_{i\sim P_e(\mathbf{s})}
		\bigl[\Delta_{e,i}f(\mathbf{s})\bigr]
		=
		\sum_i p_i w_i
		=
		\sum_i p_i y_i
		=
		G_e(\mathbf{s}).
		$
		
		We next prove that, under the proportional Top-2 rule, there exist \(A\) and \(Q\) such that \(A+Q=\sqrt{2}\).
		
		\begin{lemma}\label{lem:proportional-top2-certificates}
			For the proportional Top-2 rule, there is a \((C,A,Q)\)-certificate with
			\[
			C=1,\qquad
			A=\frac{\sqrt{2}-1}{2},\qquad
			Q=\frac{1+\sqrt{2}}{2}.
			\]
			Consequently, \(A+Q=\sqrt{2}\).
		\end{lemma}
		
		\begin{proof}
			Under the normalization \(y_1=1\) and \(y_2=z\), by~\eqref{proportional}, we have
			$
			G(z)=p+(1-p)z=\frac{1+z^2}{1+z}.
			$
			
			We first prove \(C=1\). By Lemma~\ref{lem:local-upper-envelopes}\emph{(i)}, the replacement loss is at most
			\[
			\max\{2(1-p),2pz\}=\frac{2z}{1+z}\le G(z),
			\]
			where the last inequality follows from
			$
			G(z)-\frac{2z}{1+z}
			=
			\frac{(1-z)^2}{1+z}\ge 0.
			$
			Thus \(C=1\) is valid.
			
			We next prove \(A=(\sqrt{2}-1)/2\). Since \(p=1/(1+z)\ge 1/2\), Lemma~\ref{lem:local-upper-envelopes}\emph{(ii)} gives that the pure-addition loss is at most \((2p-1)z=z(1-z)/(1+z)\). Dividing by \(G(z)\), we obtain
			\[
			r(z):=\frac{z(1-z)}{1+z^2},
			\qquad
			r'(z)=\frac{1-2z-z^2}{(1+z^2)^2}.
			\]
			The numerator has a unique zero in \([0,1]\), namely \(z_*=\sqrt{2}-1\). Hence \(r\) is maximized at \(z_*\), and
			$
			r(z_*)=\frac{\sqrt{2}-1}{2}.
			$
			Therefore \(A=(\sqrt{2}-1)/2\) is valid. Moreover, Instance A in Lemma~\ref{lem:sharpness-atom} attains this ratio at \(z=z_*\), so \(A\) cannot be decreased for the proportional Top-2 rule.
			
			Finally, we prove \(Q=(1+\sqrt{2})/2\). By Lemma~\ref{lem:local-upper-envelopes}\emph{(iii)}, the smallest valid uniform \(Q\) is
			\[
			\max_{0\le z\le 1}\frac{1}{G(z)}
			=
			\max_{0\le z\le 1}\frac{1+z}{1+z^2}.
			\]
			Let \(q(z):=(1+z)/(1+z^2)\). Then
			$
			q'(z)=(1-2z-z^2)/(1+z^2)^2,
			$
			so \(q\) is also maximized at \(z_*=\sqrt{2}-1\), and
			$
			q(z_*)=\frac{1+\sqrt{2}}{2}.
			$
			Thus \(Q=(1+\sqrt{2})/2\) is valid. By the necessity part of Lemma~\ref{lem:local-upper-envelopes}\emph{(iii)}, this \(Q\) cannot be decreased. Combining the bounds gives
			$
			A+Q
			=
			\frac{\sqrt{2}-1}{2}
			+
			\frac{1+\sqrt{2}}{2}
			=
			\sqrt{2}.
			$
		\end{proof}

		\section{Single matroid constraints}\label{sec4}
		This section gives a randomized greedy algorithm based on the proportional Top-2 rule for \textit{matroid-constrained non-monotone \(k\)-submodular maximization} (M\(k\)SM) with \(k \ge 2\), summarized in Algorithm~\ref{alo1}, and proves that it achieves a \((\sqrt{2}-1)\)-approximation.
		
		Let \(\mathcal{M}=(V,\mathcal{P})\) be a matroid of rank \(B\). The main idea of Algorithm~\ref{alo1} is as follows. Suppose that the algorithm is at iteration \(t\), and let \(\mathbf{s}^{t-1}\) be the current partial solution, initialized as \(\mathbf{s}^0=\mathbf{0}\). Throughout the algorithm, we maintain that \(S_{t-1}:=\operatorname{supp}(\mathbf{s}^{t-1})\) is an independent set. Using the proportional Top-2 rule, we assign an exact score \(G_e\) to every element \(e\in V\setminus S_{t-1}\). We then compute a maximum-score base \(M_t\) of the contraction \(\mathcal{M}/S_{t-1}\).  By the Rado--Edmonds greedy theorem for matroids~\cite{edmonds1971matroids}, \(M_t\) can be found greedily: scan the elements of \(V\setminus S_{t-1}\) in nonincreasing order of score and add \(e\) to the current set \(T\) whenever
		$
		S_{t-1}\cup T\cup\{e\}\in\mathcal{P}.
		$ Finally, the algorithm chooses an element uniformly at random from \(M_t\), assigns it a label according to the proportional Top-2 rule, and adds it to \(\mathbf{s}^{t-1}\). Since \(\mathcal{M}\) has rank \(B\), the algorithm terminates after at most \(B\) iterations and returns a feasible solution \(\mathbf{s}^B\).
		
		\begin{algorithm}[H]
			\caption{Matroid proportional top-2 randomized greedy}
			\label{alo1}
			\begin{algorithmic}[1]
				\State \textbf{Input:} value oracle \(\mathcal{O}_f\); independence oracle for the matroid \(\mathcal{M}=(V,\mathcal{P})\); rank \(B\).
				\State \textbf{Output:} a feasible solution \(\mathbf{s}^B\).
				\State \(\mathbf{s}^0\gets \mathbf{0}\), \(S_0\gets\emptyset\).
				\For{\(t=1,2,\ldots,B\)}
				\ForAll{\(e\in V\setminus S_{t-1}\)}
				\State Compute \(w_{e,i}\gets f(\mathbf{s}^{t-1}+i\mathbf{1}_e)-f(\mathbf{s}^{t-1})\) for all \(i\in[k]\).
				\State Let \(\alpha_e,\beta_e\) be two labels with the largest values of \(w_{e,i}\), with \(w_{e,\alpha_e}\ge w_{e,\beta_e}\).
				\State \(y_{e,1}\gets\max\{w_{e,\alpha_e},0\}\), \(y_{e,2}\gets\max\{w_{e,\beta_e},0\}\).
				\If{\(y_{e,1}=0\)}
				\State \(p_{e,1}\gets1\), \(p_{e,2}\gets0\), \(G_e\gets0\).
				\Else
				\State \(p_{e,1}\gets y_{e,1}/(y_{e,1}+y_{e,2})\), \(p_{e,2}\gets y_{e,2}/(y_{e,1}+y_{e,2})\).
				\State \(G_e\gets p_{e,1}y_{e,1}+p_{e,2}y_{e,2}
				=(y_{e,1}^2+y_{e,2}^2)/(y_{e,1}+y_{e,2})\).
				\EndIf
				\EndFor
				\State Choose any
				\[
				M_t\in\arg\max
				\left\{
				\sum_{e\in M}G_e:
				M \text{ is a base of } \mathcal{M}/S_{t-1}
				\right\}.
				\]
				\State Choose an element \(e_t\) uniformly at random from \(M_t\).
				\State Choose a label \(i_t\) by setting
				\[
				\Pr[i_t=\alpha_{e_t}]=p_{e_t,1},
				\qquad
				\Pr[i_t=\beta_{e_t}]=p_{e_t,2}.
				\]
				\State Update \(\mathbf{s}^t\gets \mathbf{s}^{t-1}+i_t\mathbf{1}_{e_t}\), \(S_t\gets S_{t-1}\cup\{e_t\}\).
				\EndFor
				\State \Return \(\mathbf{s}^B\).
			\end{algorithmic}
		\end{algorithm}
		
		\FloatBarrier
		
		Algorithm~\ref{alo1} has the following approximation guarantee and query complexity.
		
		\begin{theorem}\label{mat_approx}
			Algorithm~\ref{alo1} outputs a solution \(\mathbf{s}^B\) satisfying
			$
			\mathbb{E}\bigl[f(\mathbf{s}^B)\bigr]
			\ge
			(\sqrt{2}-1)f(\mathbf{o}),
			$
			using at most \(O(n^2k)\) calls to the value oracle \(\mathcal{O}_f\), where \(\mathbf{o}\) is an optimal solution to MkSM.
		\end{theorem}
		
		To prove Theorem~\ref{mat_approx}, we need the following facts.
		
		\begin{lemma}\label{base_opt}
			Given a matroid \(\mathcal{M}\) of rank \(B\), there exists an optimal solution \(\mathbf{o}\) to MkSM such that \(\operatorname{supp}(\mathbf{o})\) is a base of \(\mathcal{M}\).
		\end{lemma}
		
		\begin{proof}
			Pick an arbitrary optimal solution \(\mathbf{o}\). If \(\operatorname{supp}(\mathbf{o})\) is not a base of \(\mathcal{M}\), then there exists an element \(e\notin \operatorname{supp}(\mathbf{o})\) such that
			$
			\operatorname{supp}(\mathbf{o})\cup\{e\}\in\mathcal{P}.
			$
			If \(\Delta_{e,i}f(\mathbf{o})<0\) for all labels \(i\in[k]\), then for any two distinct labels \(i,j\), we have
			$
			\Delta_{e,i}f(\mathbf{o})+\Delta_{e,j}f(\mathbf{o})<0,
			$
			which contradicts pairwise monotonicity. Therefore, there exists a label \(i\) such that \(\Delta_{e,i}f(\mathbf{o})\ge 0\). Assigning this label to \(e\) does not decrease the function value, and the new support remains independent. Repeating this process increases the support size by \(1\) in each step, and after at most \(B\) steps, the support becomes a base.
		\end{proof}
		
		\begin{lemma}[\cite{brualdi1969comments}]\label{lem:base-exchange}
			Let \(M\) and \(O\) be two bases of a matroid \(\mathcal{M}\). There exists a bijection
			\[
			\pi:M\to O
			\]
			such that:
			\begin{enumerate}
				\item[(i)] if \(e\in M\cap O\), then \(\pi(e)=e\);
				\item[(ii)] for every \(e\in M\), the set \(O-\pi(e)+e\) \footnote{Here, we employ the shorthand \(O+e\) for \(O \cup \{e\}\) and  \(O-e\) for \(O\setminus\{e\}\).} is a base of \(\mathcal{M}\).
			\end{enumerate}
		\end{lemma}
		
		By Lemma~\ref{lem:base-exchange}, we obtain the following observation.
		
		\begin{observation}\label{observation1}
			Let \(M\) be a maximum-weight base with respect to a nonnegative weight function \(G\), and let \(O\) be any base. Then clearly
			$
			\sum_{e\in M}G_e
			\ge
			\sum_{r\in O}G_r.
			$
			Moreover, if \(\pi:M\to O\) is the base-exchange bijection from Lemma~\ref{lem:base-exchange}, then subtracting the common weight \(\sum_{e\in M\cap O}G_e\) from both sides gives
			$
			\sum_{e\in M\setminus O}G_e
			\ge
			\sum_{e\in M\setminus O}G_{\pi(e)}.
			$
		\end{observation}
		
		\begin{proof}[Proof of Theorem~\ref{mat_approx}]
			By Lemma~\ref{base_opt}, fix an optimal solution \(\mathbf{o}\) whose support is a base. In the analysis, let
			$
			S_t:=\operatorname{supp}(\mathbf{s}^t),
			$
			and maintain a residual base \(O_t\subseteq \operatorname{supp}(\mathbf{o})\) and the hybrid solution
			$
			\mathbf{h}^t=\mathbf{h}(\mathbf{s}^t,\mathbf{o}^t),
			$
			such that \(S_t\cup O_t\) is a base of the original matroid. Here \(\mathbf{o}^t\) is the restriction of \(\mathbf{o}\) to \(O_t\), namely \(\mathbf{o}^t_e=\mathbf{o}_e\) for \(e\in O_t\), and \(\mathbf{o}^t_e=0\) for \(e\notin O_t\). Initially, \(O_0=\operatorname{supp}(\mathbf{o})\) and \(\mathbf{h}^0=\mathbf{o}\). At the end, \(O_B=\emptyset\), and hence \(\mathbf{h}^B=\mathbf{s}^B\).
			
			Fix the history \(\mathcal{F}_{t-1}\) before iteration \(t\), and write
			\[
			\mathbf{s}=\mathbf{s}^{t-1},\qquad
			S=S_{t-1},\qquad
			O=O_{t-1},\qquad
			\mathbf{h}=\mathbf{h}^{t-1}.
			\]
			Let \(M=M_t\) be the maximum-score base chosen by Algorithm~\ref{alo1} in the contraction \(\mathcal{M}/S\). Note that both \(M\) and \(O\) are bases of \(\mathcal{M}/S\). By Lemma~\ref{lem:base-exchange}, there exists a bijection \(\pi:M\to O\) such that \(O-\pi(e)+e\) is a base of \(\mathcal{M}/S\) for every \(e\in M\). The algorithm chooses \(e\) uniformly at random from \(M\). Once \(e\) is chosen, the analysis sets
			$
			O'=O-\pi(e).
			$
			Since \(O-\pi(e)+e\) is a base of \(\mathcal{M}/S\), we have
			\[
			S_t\cup O_t
			=
			S\cup\{e\}\cup O'
			=
			S\cup (O-\pi(e)+e)
			\]
			is still a base of the original matroid, and hence the maintained residual base can be used in the next iteration.
			
			We first compute the expected gain of the algorithmic solution. Since the proportional Top-2 rule is an exact-score rule, conditional on choosing \(e\), we have
			$
			\mathbb{E}\bigl[f(\mathbf{s}+i\mathbf{1}_e)-f(\mathbf{s})\mid e\bigr]=G_e(\mathbf{s}).
			$
			Since \(e\) is chosen uniformly from \(M\),
			\begin{equation}\label{eq:matroid-gain}
				g:=
				\mathbb{E}\bigl[f(\mathbf{s}^t)-f(\mathbf{s}^{t-1})\mid\mathcal{F}_{t-1}\bigr]
				=
				\frac{1}{|M|}\sum_{e\in M}G_e(\mathbf{s}),
			\end{equation}
			where \(|M|=B-t+1\).
			
			We next control the hybrid loss. Decompose \(M\) into \(I=M\cap O\) and \(X=M\setminus O\). If \(e\in I\), then \(\pi(e)=e\). Let \(\ell\) be the  label of \(e\) in \(\mathbf{h}\). Since \(\mathbf{s}\preceq \mathbf{h}_{-e}\), the \(C\)-replacement bound gives
			\[
			\mathbb{E}\bigl[f(\mathbf{h})-f(\mathbf{h}^t)\mid e,\mathcal{F}_{t-1}\bigr]
			=
			\Delta_{e,\ell}f(\mathbf{h}_{-e})
			-
			\mathbb{E}_i\bigl[\Delta_{e,i}f(\mathbf{h}_{-e})\bigr]
			\le
			C G_e(\mathbf{s}).
			\]
			If \(e\in X\), since \(\pi\) fixes \(M\cap O\), we have \(r=\pi(e)\in O\setminus M\). We split the hybrid loss into two parts: first adding \(e\) with a random label, and then deleting \(r\). The first part is bounded by the \(A\)-addition bound:
			\[
			f(\mathbf{h})-\mathbb{E}_i\bigl[f(\mathbf{h}+i\mathbf{1}_e)\bigr]
			=
			-\mathbb{E}_i\bigl[\Delta_{e,i}f(\mathbf{h})\bigr]
			\le
			A G_e(\mathbf{s}).
			\]
			For the second part, for every random label \(i\),
			$
			f(\mathbf{h}+i\mathbf{1}_e)-f(\mathbf{h}_{-r}+i\mathbf{1}_e)
			=
			\Delta_{r,\mathbf{o}_r}f(\mathbf{h}_{-r}+i\mathbf{1}_e).
			$
			Since \(\mathbf{s}\preceq \mathbf{h}_{-r}+i\mathbf{1}_e\), orthant submodularity gives
			$
			\Delta_{r,\mathbf{o}_r}f(\mathbf{h}_{-r}+i\mathbf{1}_e)
			\le
			\Delta_{r,\mathbf{o}_r}f(\mathbf{s}).
			$
			Taking expectation over the random label \(i\), and then using the \(Q\)-deletion domination, we get
			\[
			\mathbb{E}_i\bigl[\Delta_{r,\mathbf{o}_r}f(\mathbf{h}_{-r}+i\mathbf{1}_e)\bigr]
			\le
			\Delta_{r,\mathbf{o}_r}f(\mathbf{s})
			\le
			QG_r(\mathbf{s})
			=
			QG_{\pi(e)}(\mathbf{s}).
			\]
			Therefore, when \(e\in X\),
			\[
			\mathbb{E}\bigl[f(\mathbf{h})-f(\mathbf{h}^t)\mid e,\mathcal{F}_{t-1}\bigr]
			\le
			A G_e(\mathbf{s})+QG_{\pi(e)}(\mathbf{s}).
			\]
			Taking expectation over the uniformly chosen \(e\in M\), we obtain
			\begin{equation}\label{eq:matroid-loss}
				\begin{aligned}
					L
					&:=
					\mathbb{E}\bigl[f(\mathbf{h}^{t-1})-f(\mathbf{h}^t)\mid\mathcal{F}_{t-1}\bigr] \\
					&\le
					\frac{1}{|M|}
					\left(
					C\sum_{e\in I}G_e(\mathbf{s})
					+
					A\sum_{e\in X}G_e(\mathbf{s})
					+
					Q\sum_{e\in X}G_{\pi(e)}(\mathbf{s})
					\right).
				\end{aligned}
			\end{equation}
			By Observation~\ref{observation1}, we have
			$
			\sum_{e\in X}G_{\pi(e)}(\mathbf{s})
			\le
			\sum_{e\in X}G_e(\mathbf{s}).
			$
			Substituting this into~\eqref{eq:matroid-loss}, we get
			\[
			L
			\le
			\frac{1}{|M|}
			\left(
			C\sum_{e\in I}G_e(\mathbf{s})
			+
			(A+Q)\sum_{e\in X}G_e(\mathbf{s})
			\right)
			\le
			\frac{\eta}{|M|}\sum_{e\in M}G_e(\mathbf{s}),
			\]
			where \(\eta=\max\{C,A+Q\}\). Combining this with~\eqref{eq:matroid-gain} yields
			\begin{equation}\label{eq:matroid-loss-gain}
				L\le \eta g.
			\end{equation}
			
			Define the potential
			\[
			\Phi_t=(1-\rho)f(\mathbf{s}^t)+\rho f(\mathbf{h}^t),
			\qquad
			\rho=\frac{1}{1+\eta}.
			\]
			By~\eqref{eq:matroid-loss-gain}, conditioning on the history gives
			\[
			\mathbb{E}\bigl[\Phi_t-\Phi_{t-1}\mid\mathcal{F}_{t-1}\bigr]
			=
			(1-\rho)g-\rho L
			\ge
			(1-\rho)g-\rho\eta g
			=
			\bigl(1-\rho(1+\eta)\bigr)g
			=
			0.
			\]
			Thus \(\{\Phi_t\}_{t=0}^B\) is a submartingale. Initially, \(\mathbf{s}^0=\mathbf{0}\), \(\mathbf{h}^0=\mathbf{o}\), and \(f(\mathbf{0})=0\), so
			$
			\Phi_0=(1-\rho)f(\mathbf{0})+\rho f(\mathbf{o})=\rho f(\mathbf{o}).
			$
			Finally, \(\mathbf{h}^B=\mathbf{s}^B\), and hence \(\Phi_B=f(\mathbf{s}^B)\). Therefore,
			\[
			\mathbb{E}\bigl[f(\mathbf{s}^B)\bigr]
			=
			\mathbb{E}[\Phi_B]
			\ge
			\Phi_0
			=
			\rho f(\mathbf{o}).
			\]
			By the \((C,A,Q)\)-certificate in Lemma~\ref{lem:proportional-top2-certificates}, we have \(\eta=\sqrt{2}\), and therefore
			$
			\rho=\frac{1}{1+\sqrt{2}}=\sqrt{2}-1.
			$
			This proves the approximation guarantee. 
			
			The value-oracle complexity is \(O(n^2k)\). Indeed, in iteration \(t\), scores are computed only for elements in \(V\setminus S_{t-1}\). For each such element \(e\), the algorithm evaluates the \(k\) marginal gains
			$
			\Delta_{e,i}f(\mathbf{s}^{t-1}) \text{ for all } i \in [k],
			$
			using one reusable query for \(f(\mathbf{s}^{t-1})\) and \(k\) additional value-oracle queries. Hence the total number of value-oracle calls is
			$
			O\left(\sum_{t=1}^{B}(n-t+1)k\right)
			=
			O(Bnk)
			\le
			O(n^2k).
			$
			The greedy computation of the maximum-score base of \(\mathcal{M}/S_{t-1}\) uses \(O(n)\) independence-oracle calls per iteration, and hence \(O(Bn)\le O(n^2)\) such calls in total.
		\end{proof}
		
		\section{Single knapsack constraints}\label{sec5}
		This section gives a one-guess score-density greedy algorithm based on the proportional Top-2 rule for \textit{knapsack-constrained non-monotone \(k\)-submodular maximization} (K$k$SM) with \(k \ge 2\), summarized in Algorithms~\ref{alo2} and~\ref{alo3}, and proves that it achieves a \((\sqrt{2}-1)\)-approximation.
		
		If \(V=\emptyset\), return \(\mathbf0\). We first consider the case \(c_e\in(0,1]\) for all \(e\in V\); the preprocessing of zero-cost elements is described at the end of this section. The main idea is as follows. First, choose an arbitrary feasible singleton solution \(\mathbf{y}\) and fix it. Then Algorithm~\ref{alo3}, initialized with \(\mathbf{y}\), assigns an exact score \(G_e\) to every element \(e\in U:=V\setminus \operatorname{supp}(\mathbf{y})\) using the proportional Top-2 rule, and uses \(G_e/c_e\) as its density. It then selects an element \(a\) of maximum density. If adding \(a\) still satisfies the knapsack constraint, the algorithm assigns \(a\) a label according to the proportional Top-2 rule and adds it to the current solution. This process is repeated for at most \(n-1\) iterations, producing a density-greedy solution \(\mathbf{s}^{\mathbf{y}}\) based on \(\mathbf{y}\). Finally, the algorithm returns a solution with the largest function value among all singleton solutions \(\mathbf{y}\) and their corresponding greedy solutions \(\mathbf{s}^{\mathbf{y}}\).
		
		\begin{algorithm}[htbp]
			\caption{Knapsack one-guess greedy}
			\label{alo2}
			\begin{algorithmic}[1]
				\State \textbf{Input:} value oracle \(\mathcal{O}_f\); costs \(c_e\in(0,1]\) for all \(e\in V\); budget \(1\).
				\State \textbf{Output:} a feasible solution \(\mathbf{s}\).
				\State \(\mathcal{C}\gets\emptyset\).
				\ForAll{\(e\in V\) and \(i\in[k]\)}
				\State Let \(\mathbf{y}\) be the singleton solution that selects only \(e\) and assigns it label \(i\).
				\State Run Algorithm~\ref{alo3} with \(\mathbf{y}\) as input, and obtain a randomized greedy solution \(\mathbf{s}^{\mathbf{y}}\).
				\State \(\mathcal{C}\gets \mathcal{C}\cup\{\mathbf{y},\mathbf{s}^{\mathbf{y}}\}\).
				\EndFor
				\State Choose \(\mathbf{s}\in\arg\max\{f(\mathbf{x}):\mathbf{x}\in\mathcal{C}\}\).
				\State \Return \(\mathbf{s}\).
			\end{algorithmic}
		\end{algorithm}
		
		\begin{algorithm}[htbp]
			\caption{Proportional top-2 density greedy}
			\label{alo3}
			\begin{algorithmic}[1]
				\State \textbf{Input:} value oracle \(\mathcal{O}_f\); costs \(c_e\in(0,1]\) for all \(e\in V\); budget \(1\); a feasible singleton solution \(\mathbf{y}\).
				\State \textbf{Output:} a feasible solution \(\mathbf{s^y}\).
				\State \(\mathbf{s}\gets \mathbf{0}\), \(U\gets V\setminus\operatorname{supp}(\mathbf{y})\).
				\While{\(U\ne\emptyset\)}
				\ForAll{\(e\in U\)}
				\State Compute \(w_{e,i}\gets f((\mathbf{y} \sqcup \mathbf{s})+i\mathbf{1}_e)-f(\mathbf{y} \sqcup \mathbf{s})\) for all \(i\in[k]\).
				\State Let \(\alpha_e,\beta_e\) be two labels with the largest values of \(w_{e,i}\), with \(w_{e,\alpha_e}\ge w_{e,\beta_e}\).
				\State \(y_{e,1}\gets\max\{w_{e,\alpha_e},0\}\), \(y_{e,2}\gets\max\{w_{e,\beta_e},0\}\).
				\If{\(y_{e,1}=0\)}
				\State \(p_{e,1}\gets1\), \(p_{e,2}\gets0\), \(G_e\gets0\).
				\Else
				\State \(p_{e,1}\gets y_{e,1}/(y_{e,1}+y_{e,2})\), \(p_{e,2}\gets y_{e,2}/(y_{e,1}+y_{e,2})\).
				\State \(G_e\gets (y_{e,1}^2+y_{e,2}^2)/(y_{e,1}+y_{e,2})\).
				\EndIf
				\EndFor
				\State Choose an element \(a\in\arg\max_{e\in U} G_e/c_e\).
				\State \(U\gets U\setminus\{a\}\).
				\If{\(c(\mathbf{s}) +c(\mathbf{y})+c_a\le 1\)}
				\State Choose a label \(i\) by setting
				$
				\Pr[i=\alpha_{a}]=p_{a,1}
				\text{ and }
				\Pr[i=\beta_{a}]=p_{a,2}.
				$
				\State \(\mathbf{s}\gets \mathbf{s}+i\mathbf{1}_{a}\).
				\EndIf
				\EndWhile
				\State \(\mathbf{s}^{\mathbf{y}}\gets \mathbf{y}\sqcup\mathbf{s}\).
				\State \Return \(\mathbf{s}^\mathbf{y}\).
			\end{algorithmic}
		\end{algorithm}
		
		Algorithm~\ref{alo2} has the following approximation guarantee and query complexity. We retain this loss-certificate analysis as the baseline; Section~\ref{bellman:sec:knapsack} strengthens the guarantee to \(1/2\) for the same algorithm.
		
		\begin{theorem}\label{kna_approx}
			Algorithm~\ref{alo2} outputs a solution \(\mathbf{s}\) satisfying
			$
			\mathbb{E}\bigl[f(\mathbf{s})\bigr]
			\ge
			(\sqrt{2}-1)f(\mathbf{o}),
			$
			using at most \(O(n^3k^2)\) calls to the value oracle \(\mathcal{O}_f\), where \(\mathbf{o}\) is an optimal solution to K$k$SM.
		\end{theorem}
		
		To prove the above theorem, we first introduce the following notions.
		
		Fix a feasible singleton solution \(\mathbf{y}\) for K\(k\)SM, and define a function
		\(f_{\mathbf{y}}:(k+1)^{V_{\mathbf{y}}}\to\mathbb{R}\), where
		\(V_{\mathbf{y}}:=V\setminus\operatorname{supp}(\mathbf{y})\), by
		\[
		f_{\mathbf{y}}(\mathbf{x})=f(\mathbf{y}\sqcup\mathbf{x})-f(\mathbf{y}).
		\]
		Clearly, \(f_{\mathbf{y}}\) is still \(k\)-submodular. Notably, \(f_{\mathbf{y}}\) is not necessarily nonnegative, but the local analysis in Section~\ref{sec3} only uses orthant submodularity and pairwise monotonicity, and therefore the results of that section still apply to \(f_{\mathbf{y}}\). For convenience, we write \(g:=f_{\mathbf{y}}\).
		
		Let \(\mathbf{s}\in(k+1)^{V_{\mathbf{y}}}\) be the current partial greedy solution maintained by Algorithm~\ref{alo3}. Note that, in Algorithm~\ref{alo3}, the marginal gain
		\[
		f\bigl((\mathbf{y}\sqcup\mathbf{s})+i\mathbf{1}_e\bigr)-f(\mathbf{y}\sqcup\mathbf{s})
		=
		g(\mathbf{s}+i\mathbf{1}_e)-g(\mathbf{s}).
		\]
		Thus, for an element \(e\), the score \(G_e^f(\mathbf{y}\sqcup\mathbf{s})\) viewed under \(f\) is the same as the score \(G_e^g(\mathbf{s})\) viewed under \(g\). When there is no ambiguity, we simply write this common score as \(G_e(\mathbf{s})\).
		
		Let \(\mathbf{q}\in(k+1)^{V_{\mathbf{y}}}\) be a comparison solution used for the hybrid analysis, where
		$
		c(\mathbf{q})\le 1-c(\mathbf{y}).
		$
		For each \(e\in\operatorname{supp}(\mathbf{q})\), write \(\ell_e=\mathbf{q}_e\). Define a residual probability vector
		\[
		\lambda=(\lambda_e)_{e\in V_{\mathbf{y}}}\in[0,1]^{V_{\mathbf{y}}},
		\qquad
		\lambda_e=0
		\quad
		\text{for all } e\in V_{\mathbf{y}}\setminus\operatorname{supp}(\mathbf{q}).
		\]
		This vector is interpreted as follows: in the constructed hybrid solution, each element \(e\in\operatorname{supp}(\mathbf{q})\) independently keeps its label \(\ell_e\) from the comparison solution with probability \(\lambda_e\), and otherwise does not appear. We require
		\begin{equation*}
			\lambda_e=0
			\qquad
			\text{for every } e\in\operatorname{supp}(\mathbf{q})\cap\operatorname{supp}(\mathbf{s}),
		\end{equation*}
		so that the algorithmic solution and the residual part do not conflict on the same element. Let \(R(\lambda)\in(k+1)^{V_{\mathbf{y}}}\) denote the random residual solution generated according to \(\lambda\), and define
		\begin{equation*}
			H(\mathbf{s},\lambda)
			=
			\mathbb{E}\bigl[g(\mathbf{s}\sqcup R(\lambda))\bigr].
		\end{equation*}
		When \(\lambda_e=1\) for all \(e\in\operatorname{supp}(\mathbf{q})\) and \(\mathbf{s}=\mathbf{0}\), we have \(H(\mathbf{0},\mathbf{1})=g(\mathbf{q})\). When \(\lambda=\mathbf{0}\), we have \(H(\mathbf{s},\mathbf{0})=g(\mathbf{s})\).
		The potential \(H(\mathbf{s},\lambda)\) is used only in the analysis; the algorithm never evaluates or optimizes it.
		
		\subsection{One-step hybrid loss bound}
		
		The following lemma gives the one-step hybrid loss bound needed for the knapsack analysis. Fix an element \(a\) that has been selected and accepted by the knapsack constraint. Before sampling its random label, the analysis process deletes residual mass of total cost at most \(c_a\) from the residual probability vector \(\lambda\), giving priority to the residual mass of \(a\) itself. If the score densities of the elements corresponding to the other deleted coordinates are all no larger than the score density of \(a\), then the expected decrease in the hybrid value caused by accepting \(a\) and performing this residual deletion is at most \(\sqrt{2}G_a(\mathbf{s})\). This is precisely the local charging step that replaces one-to-one exchange under the knapsack constraint.
		
		\begin{lemma}\label{frchange}
			Condition on the history of Algorithm~\ref{alo3} up to the point where the element \(a\) has been selected and accepted by the knapsack constraint, but before its random label is sampled. Let \(\mathbf{s}\in (k+1)^{V_{\mathbf{y}}}\) denote the current partial greedy solution, and let \(\mathbf{q}\in (k+1)^{V_{\mathbf{y}}}\) be a comparison solution with residual probability vector
			\(\lambda\in[0,1]^{V_{\mathbf{y}}}\), as defined above.
			If \(a\in\operatorname{supp}(\mathbf{q})\), write \(\ell_a\) for its label in \(\mathbf{q}\). Denote by \(G_a=G_a(\mathbf{s})\) the score of \(a\), and by \(P_a(\mathbf{s})\) the proportional top-2 label distribution.
			
			Define
			\[
			\theta=
			\begin{cases}
				\lambda_a, & a\in \operatorname{supp}(\mathbf{q}),\\
				0, & a\notin \operatorname{supp}(\mathbf{q}).
			\end{cases}
			\]
			Suppose that a deletion vector
			$
			\delta=(\delta_e)_{e\in V_{\mathbf{y}}}
			$
			is chosen such that
			\[
			\delta_e\in[0,\lambda_e]
			\quad\text{for } e\in\operatorname{supp}(\mathbf{q}),
			\qquad
			\delta_e=0
			\quad\text{for } e\in V_{\mathbf{y}}\setminus\operatorname{supp}(\mathbf{q}).
			\]
			Here \(\delta_e\) represents the probability mass removed from coordinate \(e\) of \(\lambda\). Write
			\[
			D=\sum_{e\in V_{\mathbf{y}}}\delta_e c_e
			=
			\sum_{e\in\operatorname{supp}(\mathbf{q})}\delta_e c_e.
			\]
			If the following conditions hold:
			\begin{enumerate}
				\item[(i)] \(\delta_a=\theta\) and \(\theta c_a\le D\le c_a\);
				\item[(ii)] for every \(e\neq a\) with \(\delta_e>0\), the density domination
				\begin{equation}
					\frac{G_e(\mathbf{s})}{c_e}\le \frac{G_a(\mathbf{s})}{c_a}
					\label{density}
				\end{equation}
				holds,
			\end{enumerate}
			then, letting \(\lambda'=\lambda-\delta\) and \(i\sim P_a(\mathbf{s})\), we have
			\[
			H(\mathbf{s},\lambda)
			-
			\mathbb{E}_i\bigl[H(\mathbf{s}+i\mathbf{1}_a,\lambda')\bigr]
			\le
			\eta G_a(\mathbf{s}),
			\]
			where \(\eta=\sqrt{2}\).
		\end{lemma}
		
		\begin{proof}
			Throughout the proof, if \(a\notin \operatorname{supp}(\mathbf{q})\), then the label \(\ell_a\) may be specified arbitrarily, because the corresponding coefficient is \(\theta=0\), and hence it does not affect the argument.
			
			We first remove the probability mass of \(a\) from \(\lambda\). Define a new residual probability vector \(\mu\in[0,1]^{V_{\mathbf{y}}}\) by
			\[
			\mu_a=0,\qquad \mu_e=\lambda_e\quad (e\ne a).
			\]
			If \(a\notin\operatorname{supp}(\mathbf{q})\), then this simply means \(\mu=\lambda\). Since \(\delta_a=\theta\) by assumption, we have
			$
			\lambda'_a=\lambda_a-\delta_a=0
			$
			when \(a\in\operatorname{supp}(\mathbf{q})\); if \(a\notin\operatorname{supp}(\mathbf{q})\), then the coordinate of \(\lambda'\) corresponding to \(a\) is already \(0\). Therefore, under the state \(\mathbf{s}+i\mathbf{1}_a\), the random residual solutions generated from \(\mu\) and \(\lambda'\) do not contain \(a\), and hence the constructed hybrid solutions are valid.
			
			Decompose the hybrid decrease as
			\[
			H(\mathbf{s},\lambda)
			-
			\mathbb{E}_i\bigl[H(\mathbf{s}+i\mathbf{1}_a,\lambda')\bigr]
			=
			T_1+T_2,
			\]
			where
			\[
			T_1
			=
			H(\mathbf{s},\lambda)
			-
			\mathbb{E}_i\bigl[H(\mathbf{s}+i\mathbf{1}_a,\mu)\bigr],
			\]
			and
			\[
			T_2
			=
			\mathbb{E}_i\bigl[H(\mathbf{s}+i\mathbf{1}_a,\mu)\bigr]
			-
			\mathbb{E}_i\bigl[H(\mathbf{s}+i\mathbf{1}_a,\lambda')\bigr].
			\]
			
			We first estimate \(T_1\). Let \(R_{-a}(\rho)\) denote the random residual solution obtained by independently sampling elements in \(V_{\mathbf{y}}\setminus\{a\}\) according to the restriction of a residual probability vector \(\rho\): each \(e\in\operatorname{supp}(\mathbf{q})\setminus\{a\}\) appears with label \(\ell_e\) with probability \(\rho_e\), and otherwise it does not appear. Since \(\mu_e=\lambda_e\) for all \(e\ne a\), the random solutions \(R_{-a}(\lambda)\) and \(R_{-a}(\mu)\) have the same distribution, and thus can be coupled to have the same realization, denoted by \(\mathbf{r}\). Let
			$
			\mathbf{u}=\mathbf{s}\sqcup \mathbf{r}.
			$
			Since the residual vector satisfies \(\lambda_e=0\) whenever \(e\in\operatorname{supp}(\mathbf{s})\), the supports of \(\mathbf{s}\) and \(\mathbf{r}\) do not conflict, and hence
			$
			\mathbf{s}\preceq \mathbf{u}.
			$
			Moreover, since \(\mathbf{r}_a=0\) and \(a\notin\operatorname{supp}(\mathbf{s})\), we have
			$
			\mathbf{u}_a=0.
			$
			
			In the previous hybrid solution \(\mathbf{s}\sqcup R(\lambda)\), after fixing the residual part other than \(a\) as \(\mathbf{r}\), \(a\) has two possible states in \(R(\lambda)\): it appears with label \(\ell_a\) with probability \(\theta\), and it does not appear with probability \(1-\theta\). Thus the conditional expected value of the previous hybrid is
			\[
			\theta\, g(\mathbf{u}+\ell_a\mathbf{1}_a)+(1-\theta)g(\mathbf{u}).
			\]
			On the other hand, in the updated hybrid solution \((\mathbf{s}+i\mathbf{1}_a)\sqcup R(\mu)\), the element \(a\) has already been added as an algorithmic element and receives label \(i\sim P_a(\mathbf{s})\); meanwhile, \(\mu_a=0\), so \(R(\mu)\) no longer contains \(a\). Therefore, conditioned on the same \(\mathbf{r}\), the conditional expected value of the updated hybrid over the random label of \(a\) is
			\[
			\mathbb{E}_{i\sim P_a(\mathbf{s})}\bigl[g(\mathbf{u}+i\mathbf{1}_a)\bigr].
			\]
			Thus, conditioned on \(\mathbf{r}\), \(T_1\) equals
			\[
			\theta\, g(\mathbf{u}+\ell_a\mathbf{1}_a)+(1-\theta)g(\mathbf{u})
			-
			\mathbb{E}_i\bigl[g(\mathbf{u}+i\mathbf{1}_a)\bigr].
			\]
			Rewriting this expression gives
			\begin{align*}
				&\theta\, g(\mathbf{u}+\ell_a\mathbf{1}_a)+(1-\theta)g(\mathbf{u})
				-
				\mathbb{E}_i\bigl[g(\mathbf{u}+i\mathbf{1}_a)\bigr]\\
				&=\theta\bigl(g(\mathbf{u}+\ell_a\mathbf{1}_a)-g(\mathbf{u})\bigr)
				-
				\mathbb{E}_i\bigl[g(\mathbf{u}+i\mathbf{1}_a)-g(\mathbf{u})\bigr]\\
				&=\theta\,\Delta_{a,\ell_a}g(\mathbf{u})
				-
				\mathbb{E}_i\bigl[\Delta_{a,i}g(\mathbf{u})\bigr]\\
				&=\theta\Bigl(\Delta_{a,\ell_a}g(\mathbf{u})
				-\mathbb{E}_i\bigl[\Delta_{a,i}g(\mathbf{u})\bigr]\Bigr)
				+(1-\theta)\Bigl(-\mathbb{E}_i\bigl[\Delta_{a,i}g(\mathbf{u})\bigr]\Bigr).
			\end{align*}
			
			We now apply the local certificates. Since \(\mathbf{s}\preceq\mathbf{u}\) and \(\mathbf{u}_a=0\), the \(C\)-replacement bound applied to element \(a\), current state \(\mathbf{s}\), and extended state \(\mathbf{u}\) gives
			\[
			\Delta_{a,\ell_a}g(\mathbf{u})
			-
			\mathbb{E}_i\bigl[\Delta_{a,i}g(\mathbf{u})\bigr]
			\le
			C G_a(\mathbf{s}).
			\]
			Similarly, the \(A\)-addition bound gives
			\[
			-\mathbb{E}_i\bigl[\Delta_{a,i}g(\mathbf{u})\bigr]
			\le
			A G_a(\mathbf{s}).
			\]
			Therefore, for every fixed \(\mathbf{r}\),
			\begin{align*}
				&\theta\Bigl(\Delta_{a,\ell_a}g(\mathbf{u})
				-\mathbb{E}_i\bigl[\Delta_{a,i}g(\mathbf{u})\bigr]\Bigr)
				+(1-\theta)\Bigl(-\mathbb{E}_i\bigl[\Delta_{a,i}g(\mathbf{u})\bigr]\Bigr)\\
				&\le
				\theta C G_a(\mathbf{s})+(1-\theta)A G_a(\mathbf{s}).
			\end{align*}
			The right-hand side is independent of \(\mathbf{r}\). Taking expectation over \(\mathbf{r}\), we obtain
			\begin{equation}
				T_1\le \bigl(\theta C+(1-\theta)A\bigr)G_a(\mathbf{s}).
				\label{T1}
			\end{equation}
			
			We next estimate \(T_2\). Moving from \(\mu\) to \(\lambda'\) only decreases the residual mass of coordinates \(e\ne a\) in \(\mu\). Let
			\[
			E_\delta=\{e\in\operatorname{supp}(\mathbf{q}): e\ne a,\ \delta_e>0\}.
			\]
			If \(E_\delta=\emptyset\), then \(\mu=\lambda'\), and hence \(T_2=0\). Otherwise, assume
			\[
			E_\delta=\{r_1,r_2,\ldots,r_m\}.
			\]
			We remove the mass \(\delta_{r_j}\) from \(\mu_{r_j}\) one coordinate at a time, in this order. Define a sequence of residual vectors
			\[
			\nu^0,\nu^1,\ldots,\nu^m
			\]
			as follows. Set \(\nu^0=\mu\). For each \(j=1,\ldots,m\), let \(\nu^j\) be obtained from \(\nu^{j-1}\) by decreasing only coordinate \(r_j\) by \(\delta_{r_j}\), namely
			\[
			\nu^j_{r_j}=\nu^{j-1}_{r_j}-\delta_{r_j},
			\qquad
			\nu^j_e=\nu^{j-1}_e \quad (e\ne r_j).
			\]
			Then \(\nu^m=\lambda'\).
			This is valid because \(r_j\ne a\) implies \(\mu_{r_j}=\lambda_{r_j}\), and the assumption gives
			$
			0<\delta_{r_j}\le \lambda_{r_j}.
			$
			By telescoping,
			\begin{align*}
				T_2
				&=
				\mathbb{E}_i\bigl[
				H(\mathbf{s}+i\mathbf{1}_a,\nu^0)
				-
				H(\mathbf{s}+i\mathbf{1}_a,\nu^m)
				\bigr]\\
				&=
				\mathbb{E}_i\left[
				\sum_{j=1}^m
				\bigl(
				H(\mathbf{s}+i\mathbf{1}_a,\nu^{j-1})
				-
				H(\mathbf{s}+i\mathbf{1}_a,\nu^j)
				\bigr)
				\right].
			\end{align*}
			
			Fix a label \(i\) and a coordinate \(r_j\). As before, for a residual probability vector \(\rho\), let \(R_{-r_j}(\rho)\) denote the random residual solution obtained by sampling all elements except \(r_j\): each \(e\in\operatorname{supp}(\mathbf{q})\setminus\{r_j\}\) appears with label \(\ell_e\) with probability \(\rho_e\), and otherwise it does not appear. Since \(\nu^{j-1}_e=\nu^j_e\) for all \(e\ne r_j\), we couple \(R_{-r_j}(\nu^{j-1})\) and \(R_{-r_j}(\nu^j)\) so that they have the same realization, denoted by \(\mathbf{t}\). Let
			$
			\mathbf{h}=(\mathbf{s}+i\mathbf{1}_a)\sqcup \mathbf{t}.
			$
			Since \(r_j\ne a\), \(\mathbf{t}_{r_j}=0\). Also, \(r_j\notin\operatorname{supp}(\mathbf{s})\); otherwise \(\lambda_{r_j}=0\), contradicting \(\delta_{r_j}>0\). Hence
			$
			\mathbf{h}_{r_j}=0.
			$
			Moreover, \(\mathbf{h}\) extends \(\mathbf{s}\) without changing any label already assigned in \(\mathbf{s}\), and therefore
			$
			\mathbf{s}\preceq \mathbf{h}.
			$
			
			Now compare
			$
			H(\mathbf{s}+i\mathbf{1}_a,\nu^{j-1})
			\text{ and }
			H(\mathbf{s}+i\mathbf{1}_a,\nu^j).
			$
			They differ only in the probability that \(r_j\) appears. Under the former, \(r_j\) appears with probability \(\nu^{j-1}_{r_j}\); under the latter, \(r_j\) appears with probability
			\[
			\nu^j_{r_j}=\nu^{j-1}_{r_j}-\delta_{r_j}.
			\]
			Therefore, their difference is exactly the removed probability mass \(\delta_{r_j}\) times the expected marginal gain of changing \(r_j\) from absent to present with label \(\ell_{r_j}\), namely
			\begin{equation}
				H(\mathbf{s}+i\mathbf{1}_a,\nu^{j-1})
				-
				H(\mathbf{s}+i\mathbf{1}_a,\nu^j)
				=
				\delta_{r_j}\,
				\mathbb{E}_{\mathbf{t}}
				\left[
				\Delta_{r_j,\ell_{r_j}}g(\mathbf{h})
				\right].
				\label{14}
			\end{equation}
			For each fixed \(\mathbf{t}\), since \(\mathbf{s}\preceq\mathbf{h}\) and \(\mathbf{h}_{r_j}=0\), orthant submodularity gives
			\begin{equation}
				\Delta_{r_j,\ell_{r_j}}g(\mathbf{h})
				\le
				\Delta_{r_j,\ell_{r_j}}g(\mathbf{s}).
				\label{15}
			\end{equation}
			By the \(Q\)-deletion domination,
			\begin{equation}
				\Delta_{r_j,\ell_{r_j}}g(\mathbf{s})
				\le
				QG_{r_j}(\mathbf{s}).
				\label{16}
			\end{equation}
			Substituting~\eqref{15} and~\eqref{16} into~\eqref{14}, for every fixed \(i\) and \(j\), we obtain
			\begin{equation}
				H(\mathbf{s}+i\mathbf{1}_a,\nu^{j-1})
				-
				H(\mathbf{s}+i\mathbf{1}_a,\nu^j)
				\le
				\delta_{r_j}QG_{r_j}(\mathbf{s}).
				\label{17}
			\end{equation}
			The right-hand side is independent of the label \(i\). Summing~\eqref{17} over \(j=1,\ldots,m\) and then taking expectation over \(i\sim P_a(\mathbf{s})\), we get
			\begin{equation}
				T_2
				\le
				Q\sum_{j=1}^m \delta_{r_j}G_{r_j}(\mathbf{s})
				=
				Q\sum_{r\ne a}\delta_rG_r(\mathbf{s}).
				\label{18}
			\end{equation}
			
			We now use density domination to convert the right-hand side into \(G_a(\mathbf{s})\). For every \(r\ne a\) with \(\delta_r>0\), assumption~\eqref{density} gives
			\[
			\frac{G_r(\mathbf{s})}{c_r}\le \frac{G_a(\mathbf{s})}{c_a}.
			\]
			Since all costs in this section are positive, i.e., \(c_r>0\) and \(c_a>0\), this is equivalent to
			\[
			G_r(\mathbf{s})\le \frac{G_a(\mathbf{s})}{c_a}c_r.
			\]
			Therefore,
			\begin{equation}
				\sum_{r\ne a}\delta_rG_r(\mathbf{s})
				\le
				\frac{G_a(\mathbf{s})}{c_a}\sum_{r\ne a}\delta_rc_r.
				\label{19}
			\end{equation}
			On the other hand, since \(D=\sum_{e\in\operatorname{supp}(\mathbf{q})}\delta_ec_e\) and \(\delta_a=\theta\), we have
			\[
			\sum_{r\ne a}\delta_rc_r
			=
			D-\delta_ac_a
			=
			D-\theta c_a.
			\]
			Using the assumption \(D\le c_a\), we get
			\begin{equation}
				D-\theta c_a
				\le
				c_a-\theta c_a
				=
				(1-\theta)c_a.
				\label{20}
			\end{equation}
			Substituting~\eqref{19} and~\eqref{20} into~\eqref{18}, we obtain
			\begin{equation}
				T_2
				\le
				Q\frac{G_a(\mathbf{s})}{c_a}(1-\theta)c_a
				=
				Q(1-\theta)G_a(\mathbf{s}).
				\label{21}
			\end{equation}
			
			Finally, combining \(T_1\) and \(T_2\), by~\eqref{T1} and~\eqref{21},
			\begin{align*}
				H(\mathbf{s},\lambda)
				-
				\mathbb{E}_i\bigl[H(\mathbf{s}+i\mathbf{1}_a,\lambda')\bigr]
				&=
				T_1+T_2\\
				&\le
				\bigl(\theta C+(1-\theta)A\bigr)G_a(\mathbf{s})
				+
				Q(1-\theta)G_a(\mathbf{s})\\
				&=
				\bigl(\theta C+(1-\theta)(A+Q)\bigr)G_a(\mathbf{s}).
			\end{align*}
			Let \(\eta=\max\{C,A+Q\}\), since \(0\le\theta\le 1\), we have
			\[
			\theta C+(1-\theta)(A+Q)\le \eta.
			\]
			Therefore,
			\[
			H(\mathbf{s},\lambda)
			-
			\mathbb{E}_i\bigl[H(\mathbf{s}+i\mathbf{1}_a,\lambda')\bigr]
			\le
			\eta G_a(\mathbf{s}).
			\]
			By the \((C,A,Q)\)-certificate in Lemma~\ref{lem:proportional-top2-certificates}, we have \(\eta=\sqrt{2}\). This proves the lemma.
		\end{proof}
		
		\subsection{Residual-deletion process for a fixed comparison solution}
		
		Lemma~\ref{frchange} gives a one-step hybrid loss bound for any residual deletion satisfying the cost and density domination conditions. The next lemma shows that such deletions can be chosen throughout the greedy process for any fixed comparison solution \(\mathbf{q}\). Before the first element of \(\operatorname{supp}(\mathbf{q})\) is rejected by the knapsack, every remaining residual element of \(\mathbf{q}\) is still a candidate; hence the density-maximality of the accepted element gives the required density domination. A potential argument then implies that, once the accepted cost covers \(c(\mathbf{q})\), the final greedy solution captures at least a \(\rho\)-fraction of \(g(\mathbf{q})\).
		
		\begin{lemma}\label{lem:5.3}
			Fix a feasible singleton solution \(\mathbf{y}\) for K\(k\)SM. Let \(\mathbf{q}\in(k+1)^{V_{\mathbf{y}}}\) be an arbitrary comparison solution. Consider the execution of Algorithm~\ref{alo3}. Let \(\mathcal{A}_{\mathbf{q}}\) be the set of accepted steps, i.e., steps in which an element is successfully added to the knapsack, from the beginning of the algorithm until immediately before the first element in \(\operatorname{supp}(\mathbf{q})\) is rejected due to insufficient budget. If no element in \(\operatorname{supp}(\mathbf{q})\) is ever rejected, then \(\mathcal{A}_{\mathbf{q}}\) consists of all accepted steps up to the time when all elements in \(\operatorname{supp}(\mathbf{q})\) have been processed.
			
			Suppose that the total cost of the elements accepted in \(\mathcal{A}_{\mathbf{q}}\) is at least \(c(\mathbf{q})\). Then the final partial greedy solution \(\mathbf{s}^{\mathrm{fin}} \in (k+1)^{V_{\mathbf{y}}}\) maintained by Algorithm~\ref{alo3} satisfies  
			\[
			\mathbb{E}\bigl[g(\mathbf{s}^{\mathrm{fin}})\bigr]
			\ge
			\rho g(\mathbf{q}),
			\qquad
			\rho=\frac{1}{1+\sqrt{2}}=\sqrt{2}-1.
			\]
		\end{lemma}
		
		\begin{proof}
			We construct a residual-deletion process used only for the analysis. Initially, \(\mathbf{s}^0=\mathbf{0}\), and \(\lambda^0_e=1\) for all \(e\in\operatorname{supp}(\mathbf{q})\). Thus
			\begin{equation}
				H(\mathbf{s}^0,\lambda^0)=g(\mathbf{q}),
				\qquad
				g(\mathbf{s}^0)=0.
				\label{22}
			\end{equation}
			We follow the accepted steps of Algorithm~\ref{alo3}. Notably, Rejected steps whose elements are not in \(\supp(\mathbf q)\) do not
			change the greedy solution \(\mathbf s\) or the residual vector
			\(\lambda\), and hence no residual deletion is performed at such steps.
			
			Suppose that, before the \(t\)-th accepted step, the current greedy solution is \(\mathbf{s}^{t-1}\) and the residual vector is \(\lambda^{t-1}\). Define the total residual cost as
			\[
			R^{t-1}
			=
			\sum_{e\in\operatorname{supp}(\mathbf{q})}
			\lambda^{t-1}_e c_e .
			\]
			If \(R^{t-1}=0\), the analysis process stops. Otherwise, let \(a_t\) be the element accepted by the algorithm in the \(t\)-th accepted step. We remove residual mass of total cost
			$
			D^t=\min\{c_{a_t},R^{t-1}\}
			$
			from the current residual vector \(\lambda^{t-1}\), giving priority to the residual mass of \(a_t\) itself.
			
			Specifically, define
			\[
			\theta^t=
			\begin{cases}
				\lambda^{t-1}_{a_t}, & a_t\in\operatorname{supp}(\mathbf{q}),\\
				0, & a_t\notin\operatorname{supp}(\mathbf{q}).
			\end{cases}
			\]
			First set \(\delta^t_{a_t}=\theta^t\). This is feasible. Indeed, if \(D^t=R^{t-1}\), then \(D^t=R^{t-1}\ge \theta^t c_{a_t}\); if \(D^t=c_{a_t}\), then \(D^t=c_{a_t}\ge \theta^t c_{a_t}\), since \(0\le\theta^t\le1\).
			
			If, after removing the residual mass of \(a_t\), we still need to remove additional mass from the positive coordinates of \(\lambda^{t-1}\) in order to make up cost
			$
			D^t-\theta^t c_{a_t},
			$
			then such mass exists. Indeed, after removing \(a_t\), the remaining residual cost is \(R^{t-1}-\theta^t c_{a_t}\), and
			\[
			D^t-\theta^t c_{a_t}
			\le
			R^{t-1}-\theta^t c_{a_t}.
			\]
			Let \(\delta^t\) be the deletion vector constructed in this step, and define
			$
			\lambda^t=\lambda^{t-1}-\delta^t .
			$
			Then the total residual cost satisfies
			\begin{equation}\label{23}
				R^t
				=
				R^{t-1}-D^t
				=
				\max\{0,R^{t-1}-c_{a_t}\}.
			\end{equation}
			
			We next show that the deletion vector \(\delta^t\) constructed in this step satisfies the density domination condition required by Lemma~\ref{frchange}. The analysis has not yet passed the first time an element in \(\operatorname{supp}(\mathbf{q})\) is rejected due to insufficient budget. We claim that every element \(r\in\supp(\mathbf{q})\) with \(\lambda^{t-1}_r>0\) is still in the current candidate set, unless
			\(r=a_t\). Indeed, if such an element had been accepted earlier, then its residual mass would have been removed when it was accepted. If it had been rejected earlier, then it would be the first rejected element in \(\supp(q)\), a contradiction. Since the current greedy step chooses \(a_t\) to maximize
			$
			G_e(\mathbf{s}^{t-1})/c_e,
			$
			we have, for every element \(r\ne a_t\) with positive residual mass,
			\[
			\frac{G_r(\mathbf{s}^{t-1})}{c_r}
			\le
			\frac{G_{a_t}(\mathbf{s}^{t-1})}{c_{a_t}}.
			\]
			Therefore, the additional residual mass needed to reach deletion cost \(D^t\) can be chosen from elements satisfying the above density domination. By the preceding discussion of \(D^t\), this additional mass exists, and all conditions \emph{(i)} and \emph{(ii)} of Lemma~\ref{frchange} are satisfied. Write
			$g^t=G_{a_t}(\mathbf{s}^{t-1})$, then the expected hybrid decrease in this step is at most \(\eta g^t\), where $\eta =\sqrt{2}$.
			
			Define the potential
			\[
			\Phi_t=(1-\rho)g(\mathbf{s}^t)+\rho H(\mathbf{s}^t,\lambda^t),
			\]
			where \(\rho=1/(1+\eta)\). Fix the history \(\mathcal{F}_{t-1}\) before the \(t\)-th accepted step. Since the proportional Top-2 rule is an exact-score rule, the conditional expected gain of the algorithmic solution is
			\[
			\mathbb{E}\bigl[
			g(\mathbf{s}^t)-g(\mathbf{s}^{t-1})
			\mid\mathcal{F}_{t-1}
			\bigr]
			=
			g^t.
			\]
			By Lemma~\ref{frchange},
			\[
			\mathbb{E}\bigl[
			H(\mathbf{s}^t,\lambda^t)-H(\mathbf{s}^{t-1},\lambda^{t-1})
			\mid\mathcal{F}_{t-1}
			\bigr]
			\ge
			-\eta g^t.
			\]
			Thus
			\begin{align}
				\mathbb{E}\bigl[
				\Phi_t-\Phi_{t-1}
				\mid\mathcal{F}_{t-1}
				\bigr]
				&\ge
				(1-\rho)g^t-\rho\eta g^t \notag \\
				&=
				(1-\rho(1+\eta))g^t
				=
				0.
				\label{eq:potential-increment}
			\end{align}

			Since the total cost of the elements accepted in \(\mathcal{A}_{\mathbf{q}}\) is at least \(c(\mathbf{q})=R^0\), it follows from~\eqref{23} that there exists an accepted step in \(\mathcal{A}_{\mathbf{q}}\) after which the residual cost becomes zero. Let \(T\) be the first such accepted step. Then \(R^{t-1}>0\) for every \(t\le T\), so the inequality~\eqref{eq:potential-increment} applies to all steps \(t=1,\ldots,T\). Therefore,
			\begin{equation}
				\mathbb{E}[\Phi_T]\ge \Phi_0. \label{25*}
			\end{equation}
			Moreover,
			\[
			R^T
			=
			\max\left\{
			0,\,
			R^0-\sum_{t=1}^T c_{a_t}
			\right\}
			=
			0.
			\]
			Since all costs are positive in the present analysis, \(R^T=0\) implies \(\lambda^T=\mathbf{0}\). Hence
			\begin{equation}
				\Phi_T=H(\mathbf{s}^T,\mathbf{0})=g(\mathbf{s}^T). \label{26*}
			\end{equation}
			Initially, by~\eqref{22},
			\begin{equation}
				\Phi_0
				=
				(1-\rho)g(\mathbf{0})+\rho H(\mathbf{0},\mathbf{1})
				=
				\rho g(\mathbf{q}). \label{27*}
			\end{equation}
			Substituting~\eqref{26*} and~\eqref{27*} into~\eqref{25*}, we obtain
			$
			\mathbb{E}\bigl[g(\mathbf{s}^T)\bigr]
			=
			\mathbb{E}[\Phi_T]
			\ge
			\Phi_0
			=
			\rho g(\mathbf{q}).
			$
			
			After the analysis stops, if the algorithm continues to accept elements, then each remaining conditional expected marginal gain is still equal to a nonnegative exact score. Hence the total conditional expected gain in the remaining finite number of steps is nonnegative. Then
			\[
			\mathbb{E}\bigl[g(\mathbf{s}^{\mathrm{fin}})\bigr]
			\ge
			\mathbb{E}\bigl[g(\mathbf{s}^T)\bigr]
			\ge
			\rho g(\mathbf{q}).
			\]
			This proves the lemma.
		\end{proof}
		
		\subsection{Proof of Theorem~\ref{kna_approx}}
		We now complete the proof of Theorem~\ref{kna_approx}. The Lemma~\ref{lem:5.3} reduces the problem to a cost covering condition: given a singleton solution \(\mathbf{y}\), consider the greedy branch initialized at \(\mathbf{y}\). If the total cost of a comparison solution \(\mathbf{q}\) is covered by the elements accepted during the greedy process, then this branch obtains incremental value at least \(\rho f_{\mathbf{y}}(\mathbf{q})\). Let \(\mathbf{y}\) be the singleton solution of maximum value contained in the optimal solution. We take \(\mathbf{q}\) to be the remaining part  of the optimal solution after deleting \(\mathbf{y}\) and a maximum-cost element \(z\). Deleting the maximum cost element \(z\) guarantees that \(\mathbf{q}\) satisfies the cost covering condition, while the marginal gain of \(z\) itself is compensated by the value of \(\mathbf{y}\).

		\begin{proof}[Proof of Theorem~\ref{kna_approx}]
			Let \(\rho=\sqrt{2}-1\), and write \(O=\operatorname{supp}(\mathbf{o})\). If \(O=\emptyset\), the claim is trivial. If \(|O|=1\), then Algorithm~\ref{alo2} enumerates all singleton solutions and returns a feasible solution of value at least \(f(\mathbf{o})\), and hence the claim holds. In the following, assume that \(|O|\ge 2\).
			
			Choose a feasible singleton solution
			\[
			\mathbf{y}\in
			\arg\max\bigl\{
			f(\mathbf{x}):
			\mathbf{x}\preceq \mathbf{o}
			\text{ and }
			|\operatorname{supp}(\mathbf{x})|=1
			\bigr\},
			\]
			and write \(\operatorname{supp}(\mathbf{y})=\{y\}\). Since Algorithm~\ref{alo2} enumerates all feasible singleton solutions, one call to Algorithm~\ref{alo3} is made with input \(\mathbf{y}\).
			
			Choose an element of maximum cost in \(O\setminus\{y\}\):
			\[
			z\in\arg\max\{c_e:e\in O\setminus\{y\}\}.
			\]
			Assign \(z\) the same label as in \(\mathbf{o}\), and let the resulting singleton solution be \(\mathbf{z}\). Let \(\mathbf{q}\) be obtained from \(\mathbf{o}\) by deleting the elements \(y\) and \(z\), while keeping all remaining elements and their labels unchanged. If \(\operatorname{supp}(\mathbf{q})=\emptyset\), then
			\[
			f(\mathbf{o})=f(\mathbf{y})+f_{\mathbf{y}}(\mathbf{z}).
			\]
			By orthant submodularity,
			\[
			f_{\mathbf{y}}(\mathbf{z})
			=
			f(\mathbf{y}\sqcup\mathbf{z})-f(\mathbf{y})
			\le
			f(\mathbf{z})-f(\mathbf{0})
			=
			f(\mathbf{z}).
			\]
			By the choice of \(\mathbf{y}\), we have \(f(\mathbf{z})\le f(\mathbf{y})\). Therefore, \(f(\mathbf{o})\le 2f(\mathbf{y})\). Since Algorithm~\ref{alo2} keeps \(\mathbf{y}\) as a candidate solution, its output has value at least
			\[
			f(\mathbf{y})\ge \frac{1}{2}f(\mathbf{o})\ge \rho f(\mathbf{o}).
			\]
			
			It remains to consider the case \(\operatorname{supp}(\mathbf{q})\ne\emptyset\). Consider the call to Algorithm~\ref{alo3} with input \(\mathbf{y}\). Let \(\mathbf{s}\) be the current partial greedy solution immediately before the first element in \(\operatorname{supp}(\mathbf{q})\) is rejected due to insufficient budget; if no element in \(\operatorname{supp}(\mathbf{q})\) is ever rejected, let \(\mathbf{s}\) be the partial greedy solution after all elements in \(\operatorname{supp}(\mathbf{q})\) have been processed.
			
			If the first rejected element is \(e\in\operatorname{supp}(\mathbf{q})\), then
			\[
			c(\mathbf{s})+c_y+c_e>1\ge c(\mathbf{o}).
			\]
			Thus
			\begin{equation}
				c(\mathbf{s})
				>
				c(\mathbf{o})-c_y-c_e.
				\label{24}
			\end{equation}
			Since \(e\in\operatorname{supp}(\mathbf{q})\subseteq O\setminus\{y\}\), and \(z\) is an element of maximum cost in \(O\setminus\{y\}\), we have \(c_e\le c_z\). Moreover,
			\begin{equation}
				c(\mathbf{q})=c(\mathbf{o})-c_y-c_z. \label{25}
			\end{equation}
			Comparing~\eqref{24} and~\eqref{25}, we get \(c(\mathbf{s})>c(\mathbf{q})\).
			
			If no element in \(\operatorname{supp}(\mathbf{q})\) is ever rejected, then all elements in \(\operatorname{supp}(\mathbf{q})\) are feasible when processed and are accepted by the algorithm. Hence
			$
			c(\mathbf{s})\ge c(\mathbf{q}).
			$
			Thus, using \(\mathbf{q}\) as the comparison solution and applying Lemma~\ref{lem:5.3}, the final partial greedy solution \(\mathbf{s}^{\mathrm{fin}}\in(k+1)^{V_{\mathbf{y}}}\) maintained by Algorithm~\ref{alo3} satisfies
			\[
			\mathbb{E}\bigl[f_{\mathbf{y}}(\mathbf{s}^{\mathrm{fin}})\bigr]
			\ge
			\rho f_{\mathbf{y}}(\mathbf{q}).
			\]
			Let \(\mathbf{s}^{\mathbf{y}}=\mathbf{y}\sqcup \mathbf{s}^{\mathrm{fin}}\) be the final output of Algorithm~\ref{alo3}. Then
			\begin{equation}
				\mathbb{E}\bigl[f(\mathbf{s}^{\mathbf{y}})\bigr]
				=
				f(\mathbf{y})
				+
				\mathbb{E}\bigl[f_{\mathbf{y}}(\mathbf{s}^{\mathrm{fin}})\bigr]
				\ge
				f(\mathbf{y})+\rho f_{\mathbf{y}}(\mathbf{q}). \label{26}
			\end{equation}
			On the other hand, define
			\[
			m=f_{\mathbf{y}\sqcup\mathbf{q}}(\mathbf{z})
			=
			f(\mathbf{y}\sqcup\mathbf{q}\sqcup\mathbf{z})
			-
			f(\mathbf{y}\sqcup\mathbf{q}).
			\]
			Then
			\[
			f(\mathbf{o})=f(\mathbf{y})+f_{\mathbf{y}}(\mathbf{q})+m.
			\]
			By orthant submodularity,
			\[
			m=f_{\mathbf{y}\sqcup\mathbf{q}}(\mathbf{z})
			\le
			f(\mathbf{z})-f(\mathbf{0})
			=
			f(\mathbf{z}).
			\]
			By the definition of \(\mathbf{y}\), \(m\le f(\mathbf{z})\le f(\mathbf{y})\). Now compare the lower bound in~\eqref{26} with \(\rho f(\mathbf{o})\):
			\begin{align*}
				f(\mathbf{y})+\rho f_{\mathbf{y}}(\mathbf{q})-\rho f(\mathbf{o})
				&=
				f(\mathbf{y})+\rho f_{\mathbf{y}}(\mathbf{q})
				-
				\rho\bigl(f(\mathbf{y})+f_{\mathbf{y}}(\mathbf{q})+m\bigr)\\
				&=
				(1-\rho)f(\mathbf{y})-\rho m\\
				&\ge
				(1-\rho)f(\mathbf{y})-\rho f(\mathbf{y})\\
				&=
				(1-2\rho)f(\mathbf{y})
				\ge 0,
			\end{align*}
			where the last inequality follows from \(\rho=\sqrt{2}-1<1/2\). Hence
			\[
			\mathbb{E}\bigl[f(\mathbf{s}^{\mathbf{y}})\bigr]\ge \rho f(\mathbf{o}).
			\]
			
			Algorithm~\ref{alo2} returns a solution with the largest function value among all feasible singleton solutions and the corresponding density-greedy solutions. Denote this output by \(\mathbf{s}^*\). Then
			\[
			\mathbb{E}\bigl[f(\mathbf{s}^*)\bigr]
			\ge
			\mathbb{E}\bigl[f(\mathbf{s}^{\mathbf{y}})\bigr]
			\ge
			\rho f(\mathbf{o})
			=
			(\sqrt{2}-1)f(\mathbf{o}).
			\]
			
			It remains to prove the query complexity. Algorithm~\ref{alo2} enumerates at most \(nk\) singleton solutions, and hence calls Algorithm~\ref{alo3} at most \(nk\) times. For a fixed feasible singleton solution, Algorithm~\ref{alo3} processes at most \(n\) elements. In each iteration, it computes \(k\) marginal gains for every element in the current candidate set. Thus the number of value-oracle queries in one call to Algorithm~\ref{alo3} is
			$
			O\left(\sum_{t=0}^{n-1}(n-t)k\right)
			=
			O(n^2k).
			$
			Multiplying by the \(nk\) calls gives a total of \(O(n^3k^2)\) value-oracle queries. The values of singleton solutions themselves require only \(O(nk)\) additional queries. This proves the theorem.
		\end{proof}
		
		It remains to explain how to handle zero-cost elements, which were excluded only to make score densities well defined and to ensure that zero residual cost implies zero residual mass.
		
		\paragraph{Preprocessing zero-cost elements.}
		If elements with \(c_e=0\) are allowed, define
		$
		Z_{\mathbf{y}}
		=
		\{e\in V_{\mathbf{y}}:c_e=0\}.
		$
		Before Algorithm~\ref{alo3} enters the density greedy phase, we preprocess the elements in \(Z_{\mathbf{y}}\) one by one. For the current solution \(\mathbf{s}\) and each \(e\in Z_{\mathbf{y}}\), we still compute the score \(G_e(\mathbf{s})\) and the label distribution \(P_e(\mathbf{s})\) according to the proportional Top-2 rule, then accept \(e\) directly and assign it a random label according to \(P_e(\mathbf{s})\). Since \(c_e=0\), these operations consume no budget. After this preprocessing step, all zero-cost elements are removed from the candidate set, and the subsequent density greedy phase is run only on positive-cost elements. Hence every density \(G_e(\mathbf{s})/c_e\) is well defined.
		
		This preprocessing phase does not change the approximation ratio. Indeed, each accepted zero-cost element \(a\) still has conditional expected gain equal to its exact score \(G_a(\mathbf{s})\ge 0\). In the hybrid analysis, if \(a\in\operatorname{supp}(\mathbf{q})\), we delete all of its residual mass:
		\[
		\delta_a=\lambda_a,
		\qquad
		\delta_e=0\quad(e\ne a).
		\]
		Since \(c_a=0\), the residual cost deleted in this step is \(0\), and no residual mass needs to be deleted from any other element \(r\ne a\). Thus the one-step hybrid loss contains only the replacement/addition part, which is bounded by the \(C\)-replacement and \(A\)-addition bounds as
		\[
		\bigl(\lambda_a C+(1-\lambda_a)A\bigr)G_a(\mathbf{s})
		\le
		\eta G_a(\mathbf{s}).
		\]
		If \(a\notin\operatorname{supp}(\mathbf{q})\), then only the addition part appears, which is also covered by the same bound. Therefore, a zero-cost step satisfies the same one-step potential increment inequality as a positive-cost accepted step. Specifically, let
		\[
		\Phi^+
		=
		(1-\rho)g(\mathbf{s}+i\mathbf{1}_a)
		+
		\rho H(\mathbf{s}+i\mathbf{1}_a,\lambda')
		\]
		be the potential after accepting the zero-cost element \(a\) with label \(i\), and let
		\[
		\Phi
		=
		(1-\rho)g(\mathbf{s})+\rho H(\mathbf{s},\lambda)
		\]
		be the potential before accepting it. By the exact score property,
		\[
		\mathbb{E}\bigl[g(\mathbf{s}+i\mathbf{1}_a)-g(\mathbf{s})\bigr]
		=
		G_a(\mathbf{s}).
		\]
		On the other hand, by the one-step hybrid loss bound above,
		\[
		H(\mathbf{s},\lambda)
		-
		\mathbb{E}\bigl[H(\mathbf{s}+i\mathbf{1}_a,\lambda')\bigr]
		\le
		\eta G_a(\mathbf{s}).
		\]
		Hence
		\begin{align*}
			\mathbb{E}[\Phi^+-\Phi]
			&\ge
			(1-\rho)G_a(\mathbf{s})-\rho\eta G_a(\mathbf{s})\\
			&=
			\bigl(1-\rho(1+\eta)\bigr)G_a(\mathbf{s})
			=
			0,
		\end{align*}
		because \(\rho=1/(1+\eta)\). Thus each zero-cost accepted step in the preprocessing phase does not decrease the potential in conditional expectation.
		
		Every zero-cost accepted step in the preprocessing phase satisfies the same nonnegative potential increment inequality. Therefore, the expected potential at the end of the preprocessing phase is at least the initial potential. After that, all remaining elements have positive cost, and \(R^t=\sum_e\lambda^t_e c_e=0\) implies \(\lambda^t=\mathbf{0}\). Hence the positive-cost part of the proof of Lemma~\ref{lem:5.3} can be applied starting from the state after the preprocessing phase, giving the same conclusion. Since the proof of Theorem~\ref{kna_approx} only uses the conclusion of Lemma~\ref{lem:5.3}, the theorem remains valid once the lemma includes this preprocessing phase. The preprocessing phase processes at most \(n\) additional elements and therefore does not change the overall \(O(n^3k^2)\) order of value-oracle queries.

		\section{1/2-Approximation via Bellman-Style Backward Induction}\label{bellman:sec:framework}
		 In this section, we use Bellman-style backward induction
to establish a $1/2$-approximation under cardinality
and knapsack constraints.
Building on the proportional Top-2 rule and the
local bounds from Section~\ref{sec3}, we first develop
a common potential framework based on the hybrid gap
and remaining resources.We then apply this framework to a maximum-score greedy
algorithm under a cardinality constraint and a
one-guess score-density greedy algorithm under a
knapsack constraint, establishing a $1/2$-approximation
in both settings.

\subsection{Top-2 rule and potential framework}\label{sec:top2-framework}
We retain the proportional Top-2 rule from
Section~\ref{sec:proportional-top2}, restated for convenience in
Algorithm~\ref{alg:top2-rule}. For an unassigned element $e$ at
state $\mathbf s_t$, let $i_1,i_2$ denote its two
highest-marginal-gain labels and set
$y_j=[\Delta_{e,i_j}f(\mathbf s_t)]_+$ for $j=1,2$,
where $[a]_+=\max\{a,0\}$.

\begin{algorithm}[tb]
\caption{Proportional Top-2 label rule}\label{alg:top2-rule}
\begin{algorithmic}[1]
\Require $f$, current solution $\mathbf s_t$, and element $e\notin\operatorname{supp}(\mathbf s_t)$
\State Query $f(\mathbf s_t+i\mathbf1_e)$ for every $i\in[k]$
\State Let $i_1,i_2$ index the two largest $\Delta_{e,i}f(\mathbf s_t)$, breaking ties by smaller label index
\State $y_j\gets[\Delta_{e,i_j}f(\mathbf s_t)]_+$ for $j=1,2$
\If{$y_1=0$}
    \State \Return $i_1$
\EndIf
\State Draw $i=i_j$ with probability $y_j/(y_1+y_2)$, $j=1,2$
\State \Return $i$
\end{algorithmic}
\end{algorithm}

Let $\mathcal F_t$ denote the history before the random label choice
at step $t$, and write $\mathbb E_t[\cdot]=\mathbb E[\cdot\mid\mathcal F_t]$.
Define the score function
\[
G_e(\mathbf s_t)=\frac{y_1^2+y_2^2}{y_1+y_2}
\quad\text{if }y_1>0,\qquad
G_e(\mathbf s_t)=0\quad\text{if }y_1=0.
\]
For the label $i$ assigned by
the rule, this score is exactly the conditional expected gain
\(
\mathbb E_t\![f(\mathbf s_t+i\mathbf1_e)-f(\mathbf s_t)]
=G_e(\mathbf s_t)\ge0.
\)

In our analysis,
fix an initial comparison solution $\mathbf h_0$ with $\mathbf s_0\preceq\mathbf h_0$. Its choice will depend on the constraint.
We will maintain a set $R_t$ of elements retained from $\mathbf h_0$ that have
not been selected by the algorithm, starting with
$R_0=\operatorname{supp}(\mathbf h_0)\setminus\operatorname{supp}(\mathbf s_0)$. The updates of $R_t$ will be specified for each constraint in
Section~\ref{sec:constraint-applications}.
Write $s_{t,u}$ for the label of element $u$ in $\mathbf s_t$,
with $s_{t,u}=0$ when $u$ is unselected, and $h_{0,u}$ for its
label in $\mathbf h_0$. At each time $t$, define the comparison vector $\mathbf h_t=(h_{t,u})_{u\in V}$ by
\[
h_{t,u}=\begin{cases}
s_{t,u},&u\in\operatorname{supp}(\mathbf s_t),\\
h_{0,u},&u\in R_t,\\
0,&\text{otherwise}.
\end{cases}
\]
Thus $\mathbf h_t$ combines the algorithm's current assignments with
the original assignments of the retained elements.
Write $D_t=f(\mathbf h_t)-f(\mathbf s_t)$.

The next lemma relates $D_t$ to the scores
of the retained elements. 
It also bounds the conditional expectation
of this difference after an accepted element replaces a retained comparison label,
is added outside the comparison, or is added while a retained
comparison element is deleted. 

\begin{lemma}\label{lem:comparison-interfaces}
At every step,
\begin{equation}\label{eq:comparison-aggregate}
D_t\le\sum_{u\in R_t}\Delta_{u,h_{t,u}}f(\mathbf s_t)
 \le \tfrac{1+\sqrt2}{2}\sum_{u\in R_t}G_u(\mathbf s_t).
\end{equation}
Suppose an unselected element $e$  receives a label
$i$ from the Top-2 rule at $\mathbf s_t$ so that
$\mathbf s_{t+1}=\mathbf s_t+i\mathbf1_e$.
Conditioning on $\mathcal F_t$, we have:
\begin{enumerate}
\item[(i)] If $e\in R_t$ and $R_{t+1}=R_t\setminus\{e\}$, then
$\mathbb E_t[D_{t+1}]\ge D_t-2G_e(\mathbf s_t)$.
\item[(ii)] If $e\notin R_t$ and $R_{t+1}=R_t$, then
$\mathbb E_t[D_{t+1}]\ge D_t-\tfrac{1+\sqrt2}{2}G_e(\mathbf s_t)$.
\item[(iii)] If $e\notin R_t$ and $R_{t+1}=R_t\setminus\{u\}$, where
$u\in R_t$ is fixed before drawing $i$ and
$G_u(\mathbf s_t)\le G_e(\mathbf s_t)$, then
$\mathbb E_t[D_{t+1}]\ge D_t-(1+\sqrt2)G_e(\mathbf s_t)$.
\end{enumerate}
\end{lemma}

\begin{proof}
Fix any unselected $e$, with Top-2 labels $i_1,i_2$ and
their positive-part values $y_1\ge y_2$. The case $y_1=0$
is trivial, since all relevant marginal gains vanish.
For $y_1>0$, the identity
\[
\tfrac{1+\sqrt2}{2}(y_1^2+y_2^2)-y_1(y_1+y_2)
=\tfrac{1+\sqrt2}{2}\bigl(y_2-(\sqrt2-1)y_1\bigr)^2\ge0
\]
gives $\Delta_{e,i}f(\mathbf s_t)\le y_1
\le \tfrac{1+\sqrt2}{2}G_e(\mathbf s_t)$ for every label $i$.
Adding the elements of $R_t$ to $\mathbf s_t$ at their
retained labels,  orthant submodularity bounds
each gain by its marginal at $\mathbf s_t$. The preceding
current-marginal bound then gives
\[
D_t=f(\mathbf h_t)-f(\mathbf s_t)
\le\sum_{u\in R_t}\Delta_{u,h_{t,u}}f(\mathbf s_t)
\le \tfrac{1+\sqrt2}{2}\sum_{u\in R_t}G_u(\mathbf s_t).
\]

For an accepted element $e$, its expected algorithmic gain is
$G_t$, so the gap changes according to
\[
\mathbb E_t[D_{t+1}]
=D_t-G_t-
 \bigl(f(\mathbf h_t)-\mathbb E_t[f(\mathbf h_{t+1})]\bigr).
\]
We bound the expected loss in comparison value for each of the three updates. First suppose that $e\in R_t$,
write $p_j=y_j/(y_1+y_2)$ for $j=1,2$, so $p_1\ge p_2$.
Orthant submodularity bounds
its $i_1$-marginal by $y_1$ and every other label marginal
by $y_2$. The expected comparison-value loss
$f(\mathbf h_t)-\mathbb E_t[f(\mathbf h_{t+1})]$ is
\[
\Delta_{e,h_{t,e}}f(\mathbf h_t-h_{t,e}\mathbf1_e)
-\sum_{j=1}^2p_j\Delta_{e,i_j}f(
  \mathbf h_t-h_{t,e}\mathbf1_e).
\]
When $h_{t,e}=i_1$ and $h_{t,e}=i_2$, respectively,
pairwise monotonicity gives
\[
\begin{aligned}
p_2\bigl(\Delta_{e,i_1}f(\mathbf h_t-i_1\mathbf1_e)
-\Delta_{e,i_2}f(\mathbf h_t-i_1\mathbf1_e)\bigr)
&\le2p_2\Delta_{e,i_1}f(\mathbf h_t-i_1\mathbf1_e)\le2p_2y_1,\\
p_1\bigl(\Delta_{e,i_2}f(\mathbf h_t-i_2\mathbf1_e)
-\Delta_{e,i_1}f(\mathbf h_t-i_2\mathbf1_e)\bigr)
&\le2p_1\Delta_{e,i_2}f(\mathbf h_t-i_2\mathbf1_e)\le2p_1y_2.
\end{aligned}
\]
Suppose $h_{t,e}\notin\{i_1,i_2\}$. The labels and
probabilities were determined at $\mathbf s_t$.
If the old-label marginal at
$\mathbf h_t-h_{t,e}\mathbf1_e$ is negative, pairwise
monotonicity makes the comparison-value loss nonpositive,
so the bound follows directly. If it is nonnegative,
decompose the weighted sum first:
\[
\begin{aligned}
\sum_{j=1}^2p_j\Delta_{e,i_j}f(
  \mathbf h_t-h_{t,e}\mathbf1_e)
={}&p_2\bigl(
  \Delta_{e,i_1}f(\mathbf h_t-h_{t,e}\mathbf1_e)
 +\Delta_{e,i_2}f(\mathbf h_t-h_{t,e}\mathbf1_e)
 \bigr)\\
&+(p_1-p_2)\Delta_{e,i_1}f(
  \mathbf h_t-h_{t,e}\mathbf1_e)\\
\ge{}&-(p_1-p_2)\Delta_{e,h_{t,e}}f(
  \mathbf h_t-h_{t,e}\mathbf1_e).
\end{aligned}
\]
Since $p_1\ge p_2$, the last step applies pairwise
monotonicity to $(i_1,i_2)$ and $(i_1,h_{t,e})$.
It follows that
\[
f(\mathbf h_t)-\mathbb E_t[f(\mathbf h_{t+1})]
\le2p_1\Delta_{e,h_{t,e}}f(
  \mathbf h_t-h_{t,e}\mathbf1_e)
\le2p_1y_2.
\]
The final inequality uses orthant submodularity and the
label ordering at $\mathbf s_t$. All four cases are bounded by $G_t=G_e(\mathbf s_t)$, since
\[
2p_2y_1=2p_1y_2=\frac{2y_1y_2}{y_1+y_2}
\le\frac{y_1^2+y_2^2}{y_1+y_2}=G_t.
\]
The gap identity gives
$\mathbb E_t[D_{t+1}]\ge D_t-2G_t$.

If $e\notin R_t$ and no element is deleted,
the label is added to $\mathbf h_t$. The same marginal bounds and
pairwise monotonicity show that pure addition loses at most
\[
-\sum_{j=1}^2p_j\Delta_{e,i_j}f(\mathbf h_t)
\le(p_1-p_2)\Delta_{e,i_2}f(\mathbf h_t)
\le\frac{y_2(y_1-y_2)}{y_1+y_2}
\le\tfrac{\sqrt2-1}{2}G_t.
\]
Hence $\mathbb E_t[D_{t+1}]\ge D_t-\tfrac{1+\sqrt2}{2}G_t$.

Finally, suppose $e\notin R_t$ and a retained element
$u$ is chosen for deletion before drawing the label of $e$.
Hence
\[
\Delta_{u,h_{t,u}}f(
  \mathbf h_t-h_{t,u}\mathbf1_u+i_j\mathbf1_e)
\le\Delta_{u,h_{t,u}}f(\mathbf s_t)
\le \tfrac{1+\sqrt2}{2}G_u(\mathbf s_t)\le \tfrac{1+\sqrt2}{2}G_t.
\]
The pure-addition bound controls the insertion term, while
the inequality above controls deletion. Thus
\[
\begin{aligned}
f(\mathbf h_t)-\mathbb E_t[f(\mathbf h_{t+1})]
={}&-\sum_{j=1}^2p_j\Delta_{e,i_j}f(\mathbf h_t)
+\sum_{j=1}^2p_j\Delta_{u,h_{t,u}}f(
  \mathbf h_t-h_{t,u}\mathbf1_u+i_j\mathbf1_e)\\
\le{}&\tfrac{\sqrt2-1}{2}G_t+\tfrac{1+\sqrt2}{2}G_t=(\sqrt2)G_t.
\end{aligned}
\]
The gap identity now gives
\[
\mathbb E_t[D_{t+1}]
\ge D_t-G_t-(\sqrt2)G_t
=D_t-(1+\sqrt2)G_t.
\]
\end{proof}

We next state a potential inequality for the sequences
$\mathbf s_t$ and $\mathbf h_t$.
Let $N$ be a fixed upper bound on the number of steps, counting each
candidate considered (no matter whether accepted or rejected). After a run ends,
keep $\mathbf s_t$ and $\mathbf h_t$ fixed until $N$, so that
$\mathbf s_N$ always denotes its final solution. 

For $t<N$, write
$G_t=\mathbb E_t[f(\mathbf s_{t+1})-f(\mathbf s_t)]$
for the conditional expected gain. If an element $e$ is accepted,
then $G_t=G_e(\mathbf s_t)$. For a rejected step or a step
after termination, $G_t=0$.
For every step $t$, 
choose nonnegative weights $\lambda_t$ based on $\mathcal F_t$,  and consider the \emph{potential}
$f(\mathbf s_t)+\lambda_t[D_t]_+$.
The following proposition gives conditions under which this potential
is a lower bound on the conditional expected final value. 

\begin{proposition}\label{prop:potential-analysis}
Let $\lambda_N[D_N]_+=0$. Suppose  that, at every history with
$t<N$ and $D_t>0$, either $G_t\ge\lambda_tD_t$, or there are
$b_t>0$ and $\lambda_{t+1}$, both determined before
the random label choice at step $t$, such that
\begin{equation}\label{eq:continuation-condition}
\mathbb E_t[D_{t+1}]\ge D_t-b_tG_t,\qquad
G_t+\lambda_{t+1}[D_t-b_tG_t]_+\ge\lambda_tD_t.
\end{equation}
Then, for every $0\le t\le N$, almost surely,
\begin{equation}\label{eq:conditional-guarantee}
\mathbb E_t[f(\mathbf s_N)]
\ge f(\mathbf s_t)+\lambda_t[D_t]_+.
\end{equation}
\end{proposition}
\begin{proof}
We prove \eqref{eq:conditional-guarantee}
by Bellman-style backward induction. At $t=N$, the claim follows from
$\lambda_N[D_N]_+=0$. Suppose it holds at $t+1$, where $t<N$.
The tower property gives
\[
\mathbb E_t[f(\mathbf s_N)]=
\mathbb E_t\!\left[\mathbb E_{t+1}[f(\mathbf s_N)]\right]
\ge\mathbb E_t\!\left[
 f(\mathbf s_{t+1})+\lambda_{t+1}[D_{t+1}]_+\right].
\]
Since $\lambda_{t+1}[D_{t+1}]_+\ge0$, we may drop this term
from the preceding expectation in the following two cases.
If $D_t\le0$, objective nondecrease gives
\[
\mathbb E_t[f(\mathbf s_N)]
\ge\mathbb E_t[f(\mathbf s_{t+1})]
\ge f(\mathbf s_t)
=f(\mathbf s_t)+\lambda_t[D_t]_+.
\]
If $D_t>0$ and $G_t\ge\lambda_tD_t$, then
\[
\mathbb E_t[f(\mathbf s_N)]
\ge\mathbb E_t[f(\mathbf s_{t+1})]
=f(\mathbf s_t)+G_t
\ge f(\mathbf s_t)+\lambda_tD_t.
\]
Otherwise $D_t>0$ and $G_t<\lambda_tD_t$.
Since $\lambda_{t+1}$ and $b_t$ are determined before the random
label choice at step $t$,
Jensen's inequality and \eqref{eq:continuation-condition} give
\[
\begin{aligned}
\mathbb E_t[f(\mathbf s_N)]
&\ge f(\mathbf s_t)+G_t
 +\lambda_{t+1}\mathbb E_t\bigl[[D_{t+1}]_+\bigr]\ge f(\mathbf s_t)+G_t
 +\lambda_{t+1}[\mathbb E_t[D_{t+1}]]_+\\
&\ge f(\mathbf s_t)+G_t
 +\lambda_{t+1}[D_t-b_tG_t]_+\ge f(\mathbf s_t)+\lambda_tD_t.
\end{aligned}
\]
Thus the induction holds at $t$, completing the proof.
\end{proof}

For a fixed initial state, the proposition gives
$\mathbb E[f(\mathbf s_N)]\ge f(\mathbf s_0)+\lambda_0[D_0]_+$.
With the initial comparison fixed, we choose weights satisfying the
proposition's conditions to make $\lambda_0$ as large as possible.
In the special case $\mathbf s_0=\mathbf0$ and $\mathbf h_0$ optimal,
$f(\mathbf0)\ge0$ turns this into a $\lambda_0$-approximation.


\subsection{Applications to cardinality and knapsack constraints}\label{sec:constraint-applications}
In this section, we apply Proposition~\ref{prop:potential-analysis}
to non-monotone $k$-submodular maximization under cardinality and knapsack constraints and present the
 combinatorial greedy algorithms. 
In both analyses, we use the following convention for
choosing the weights. For $t<N$,
\begin{equation}\label{eq:weigh}
D_t\le0\quad\text{or}\quad G_t\ge\lambda_tD_t
\quad\Longrightarrow\quad
\lambda_{t+1}=\cdots=\lambda_N=0.
\end{equation}

\subsubsection{Cardinality constraint}\label{sec:cardinality}
To maximize $f(\mathbf s)$ under $|\operatorname{supp}(\mathbf s)|\le B$,
Algorithm~\ref{alg:cardinality} starts from $\mathbf s_0=\mathbf0$.
At every step $t$, it evaluates the score $G_e(\mathbf s_t)$, the expected
Top-2 marginal gain, for each unselected element and chooses
one of maximum score. It assigns the selected element a label
using the Top-2 rule and repeats this procedure for $B$ iterations.

\begin{algorithm}[H]
\caption{Cardinality random-greedy algorithm}\label{alg:cardinality}
\begin{algorithmic}[1]
\Require $f$, $V$, $k$, and $B$
\State $\mathbf s_0\gets\mathbf0$
\For{$t=0,\ldots,B-1$}
  \State Under the Top-2 rule, compute $G_e(\mathbf s_t)=\mathbb E_t\!\left[f(\mathbf s_t+i\mathbf1_e)-f(\mathbf s_t)\right]$ for each
         $e\notin\operatorname{supp}(\mathbf s_t)$
  \State $e\gets\arg\max_{u\notin\operatorname{supp}(\mathbf s_t)}
         G_u(\mathbf s_t)$
  \State $i\gets$ label returned by the Algorithm~\ref{alg:top2-rule}
         at $\mathbf s_t$ for $e$
  \State $\mathbf s_{t+1}\gets\mathbf s_t+i\mathbf1_e$
\EndFor
\State \Return $\mathbf s_B$
\end{algorithmic}
\end{algorithm}

By pairwise monotonicity, we can take an optimum $\mathbf o$
with $|\operatorname{supp}(\mathbf o)|=B$. Set
$R_t=\operatorname{supp}(\mathbf o)\setminus
\operatorname{supp}(\mathbf s_t)$. The comparison vector
$\mathbf h_t$ is obtained from $\mathbf s_t$ by adding each element
of $R_t$ with its label in $\mathbf o$. In particular,
$\mathbf h_0=\mathbf o$.
Define the weight function
\begin{equation}\label{eq:resource-weight}
\lambda(r,m)=\frac{r}{\tfrac{1+\sqrt2}{2}m+\tfrac{3-\sqrt2}{2}r}
\quad(0\le r\le m,\ m>0),\qquad
\lambda(0,0)=\frac12.
\end{equation}
Before the weights are set to zero according to~\eqref{eq:weigh}, use
$\lambda_t=\lambda(B-t,|R_t|)$.

\begin{lemma}\label{lem:resource-weight}
The weight in~\eqref{eq:resource-weight} lies in $[0,1/2]$.
For $0<c<r\le m$, pure addition changes $(r,m)$ to
$(r-c,m)$ with $b_t=\tfrac{1+\sqrt2}{2}$, while replacement of a retained
comparison label changes it to $(r-c,m-c)$ with $b_t=2$. The respective
weight relations are
\begin{equation}\label{eq:resource-updates}
\begin{aligned}
\lambda(r,m)&\le\lambda(r-c,m)
 +\frac{c}{\tfrac{1+\sqrt2}{2}m}\bigl(1-\tfrac{1+\sqrt2}{2}\lambda(r-c,m)\bigr),\\
\lambda(r,m)&\le\lambda(r-c,m-c)
 +\frac{c}{\tfrac{1+\sqrt2}{2}m}\bigl(1-2\lambda(r-c,m-c)\bigr).
\end{aligned}
\end{equation}
If $m>0$ and $c\ge r>0$, then $\lambda(r,m)\le c/(\tfrac{1+\sqrt2}{2}m)$.
\end{lemma}

\begin{proof}
For $r>0$, $\tfrac{1+\sqrt2}{2}m+\tfrac{3-\sqrt2}{2}r\ge2r$, so
$0\le\lambda(r,m)\le1/2$ and the cases with $r=0$ follow from
\eqref{eq:resource-weight}.
For pure addition, substitution into the first inequality
of~\eqref{eq:resource-updates} reduces it to
\[
\tfrac{3+2\sqrt2}{4}m^2\le
[\tfrac{1+\sqrt2}{2}m+\tfrac{3-\sqrt2}{2}r][\tfrac{1+\sqrt2}{2}m-2\tfrac{\sqrt2-1}{2}(r-c)].
\]
As $\tfrac{1+\sqrt2}{2}m-2\tfrac{\sqrt2-1}{2}(r-c)\ge\tfrac{3-\sqrt2}{2}m>0$, expansion and
$r-c\le m$ give
\[
\begin{aligned}
&[\tfrac{1+\sqrt2}{2}m+\tfrac{3-\sqrt2}{2}r][\tfrac{1+\sqrt2}{2}m-2\tfrac{\sqrt2-1}{2}(r-c)]-\tfrac{3+2\sqrt2}{4}m^2\\
&\quad\ge(r-c)[\tfrac{1+\sqrt2}{2}(4-3\tfrac{1+\sqrt2}{2})m-2\tfrac{3-\sqrt2}{2}\tfrac{\sqrt2-1}{2}(r-c)]\\
&\quad\ge(r-c)m(4-(1+\sqrt2)-\tfrac{3+2\sqrt2}{4})\ge0,
\end{aligned}
\]
where $4-(1+\sqrt2)-\tfrac{3+2\sqrt2}{4}=\tfrac34(\sqrt2-1)^2>0$.

For replacement, direct subtraction gives
\[
\begin{aligned}
\lambda(r,m)-\lambda(r-c,m-c)
&=\frac{c}{\tfrac{1+\sqrt2}{2}m+\tfrac{3-\sqrt2}{2}r}
 \bigl(1-2\lambda(r-c,m-c)\bigr)\\
&\le\frac{c}{\tfrac{1+\sqrt2}{2}m}
 \bigl(1-2\lambda(r-c,m-c)\bigr).
\end{aligned}
\]
For $m>0$ and $c\ge r>0$,
$\lambda(r,m)=r/[\tfrac{1+\sqrt2}{2}m+\tfrac{3-\sqrt2}{2}r]\le r/(\tfrac{1+\sqrt2}{2}m)\le c/(\tfrac{1+\sqrt2}{2}m)$.
\end{proof}

The following theorem gives the approximation ratio of our greedy algorithm. 

\begin{theorem}\label{thm:cardinality}\label{bellman:thm:cardinality-main}
For nonnegative $k$-submodular maximization under a cardinality constraint,
Algorithm~\ref{alg:cardinality} uses $O(Bnk)$ value queries.
Its output $\mathbf s_B$ satisfies
$\mathbb E[f(\mathbf s_B)]\ge f(\mathbf o)/2$ for an optimal
$\mathbf o$.
\end{theorem}
\begin{proof}
The algorithm makes $N=B$ selections. Suppose that the comparison
vector and weights defined above satisfy the conditions of
Proposition~\ref{prop:potential-analysis}. Initially $|R_0|=B$,
so $\lambda_0=1/2$. Since $\mathbf s_0=\mathbf0$ and
$\mathbf h_0=\mathbf o$, the proposition gives
$\mathbb E[f(\mathbf s_B)]\ge f(\mathbf o)/2$.

We now verify these conditions for all $t$. Note that the Top-2 rule never decreases
$f$, and $0\le B-t\le |R_t|$ gives
$0\le\lambda_t\le1/2$.
At $t=N$, if $R_N\ne\varnothing$,
then $\lambda_N=0$. If $R_N=\varnothing$, then
$\mathbf h_N=\mathbf s_N$ and $D_N=0$. Thus
$\lambda_N[D_N]_+=0$ in both cases, as desired.

For $t<B$, it remains to verify~\eqref{eq:continuation-condition}
when $D_t>0$ and $G_t<\lambda_tD_t$. Every element of $R_t$ is
feasible for addition to $\mathbf s_t$.
Since the algorithm selects an element of maximum score,
$G_u(\mathbf s_t)\le G_t$ for every $u\in R_t$.
Lemma~\ref{lem:comparison-interfaces} therefore gives
\[
D_t\le\tfrac{1+\sqrt2}{2}\sum_{u\in R_t}G_u(\mathbf s_t)
\le\tfrac{1+\sqrt2}{2}|R_t|G_t.
\]
Since $|R_t|\ge B-t>0$, it follows that
$G_t\ge D_t/(\tfrac{1+\sqrt2}{2}|R_t|)$.
If $B-t=1$, then $\lambda_t\le1/(\tfrac{1+\sqrt2}{2}|R_t|)$ by \eqref{eq:resource-weight},
contradicting $G_t<\lambda_tD_t$. Hence $B-t>1$.

Let $e$ be the element selected by Algorithm~\ref{alg:cardinality}
at step $t$.
If $e\in R_t$, then $R_{t+1}=R_t\setminus\{e\}$,
and we take $b_t=2$.
If $e\notin R_t$, then $R_{t+1}=R_t$,
and we take $b_t=\tfrac{1+\sqrt2}{2}$.
In both cases, Lemma~\ref{lem:comparison-interfaces} gives
$\mathbb E_t[D_{t+1}]\ge D_t-b_tG_t$.
Both $b_t$ and $\lambda_{t+1}$ are determined before the random
label choice at step $t$. Since $1-b_t\lambda_{t+1}\ge0$,
the lower bound on $G_t$ gives
\[
G_t+\lambda_{t+1}[D_t-b_tG_t]_+
\ge\Bigl(\lambda_{t+1}+\frac{1-b_t\lambda_{t+1}}{\tfrac{1+\sqrt2}{2}|R_t|}\Bigr)D_t
\ge\lambda_tD_t,
\]
where the last inequality follows from Lemma~\ref{lem:resource-weight}
with $r=B-t$, $m=|R_t|$, and $c=1$.
This verifies \eqref{eq:continuation-condition}.
\end{proof}

\subsubsection{Knapsack constraint}\label{bellman:sec:knapsack}\label{sec:knapsack}

For a unit knapsack with nonnegative element costs,
Algorithm~\ref{alg:knapsack} guesses one element
of an optimal solution and its label by enumerating all feasible pairs.
We call the run for each fixed guess $i\mathbf1_e$ a \emph{branch}.
Starting from $i\mathbf1_e$, each branch processes the zero-cost
candidates first, then repeatedly selects
an unprocessed candidate $u$ of maximum current score
density $G_u(i\mathbf1_e\sqcup\mathbf s)/c_u$.
If its addition is feasible, the Top-2
rule chooses its label and the candidate is added. Otherwise, the
candidate is discarded. After all candidates have been processed,
the algorithm maintains the best constructed solution.
Algorithms~\ref{alg:knapsack} and~\ref{alo2} describe the same
procedure; the former simply incorporates the preprocessing of
zero-cost elements explicitly into the algorithm.
Indeed, Algorithm~\ref{alg:knapsack} follows the 1-Guess Greedy
algorithm in~\cite{wang2025approximation}, replacing its greedy phase with
our random greedy procedure based on the Top-2 rule.

\begin{algorithm}[H]
\caption{Knapsack random-greedy algorithm}\label{alg:knapsack}
\begin{algorithmic}[1]
\Require $f$, $V$, $k$, costs $c_e\ge0$, and budget $1$
\State $V\gets\{e\in V:c_e\le1\}$; $\mathbf x\gets\mathbf0$
\For{each $e\in V$ and $i\in[k]$}
  \State $\mathbf s\gets\mathbf0$;
         $U\gets V\setminus\{e\}$
  \While{$U\ne\varnothing$}
    \If{some $u\in U$ has $c_u=0$}
      \State Choose such a $u$
    \Else
      \State $u\gets\arg\max_{v\in U} G_v(i\mathbf1_e\sqcup\mathbf s)/c_v$
    \EndIf
    \State $U\gets U\setminus\{u\}$
    \If{$c(i\mathbf1_e\sqcup\mathbf s)+c_u\le1$}
      \State $j\gets$ label returned by the Algorithm~\ref{alg:top2-rule}
             at $i\mathbf1_e\sqcup\mathbf s$ for $u$
      \State $\mathbf s\gets\mathbf s+j\mathbf1_u$
    \EndIf
  \EndWhile
  \State If $f(i\mathbf1_e\sqcup\mathbf s)>f(\mathbf x)$, set
         $\mathbf x\gets i\mathbf1_e\sqcup\mathbf s$
\EndFor
\State \Return $\mathbf x$
\end{algorithmic}
\end{algorithm}

For the analysis, fix an optimum $\mathbf o$, and let
\(
\mathbf y\in
\operatorname*{arg\,max}_{\substack{
\mathbf v\preceq\mathbf o:
|\operatorname{supp}(\mathbf v)|=1}}
f(\mathbf v) 
\) be the maximum-value singleton restriction of $\mathbf o$.
We analyze the branch starting from 
$\mathbf y=i\mathbf1_e$, where $i=o_e$ is the
label of $e$ in $\mathbf o$,  through the auxiliary function $g(\mathbf v)=f(\mathbf y\sqcup\mathbf v)-f(\mathbf y)$ on $V\setminus\{e\}$.

Let $z$ be a maximum-cost element of
$\operatorname{supp}(\mathbf o)\setminus\{e\}$,
and define $\mathbf h_0=\mathbf o-\mathbf y-o_z\mathbf1_z$ by removing $\mathbf y$ and $z$ from $\mathbf o$.
As in the cardinality case, form $\mathbf h_t$ by adding
$R_t=\operatorname{supp}(\mathbf h_0)\setminus
\operatorname{supp}(\mathbf s_t)$ to $\mathbf s_t$ with the labels
in $\mathbf h_0$, and set $D_t=g(\mathbf h_t)-g(\mathbf s_t)$.

We use the weight function in~\eqref{eq:resource-weight}.
Set $r_t=c(\mathbf h_0)-c(\mathbf s_t)$, the
comparison cost still to be covered, and $m_t=c(R_t)$, the cost of
the retained elements. Before the weights are set to zero according to~\eqref{eq:weigh}, we choose
$\lambda_t=\lambda(r_t,m_t)$.


\begin{theorem}\label{thm:knapsack}\label{bellman:thm:knapsack-main}
For nonnegative $k$-submodular maximization under a single knapsack constraint,
Algorithm~\ref{alg:knapsack} uses $O(n^3k^2)$ value queries
and returns a feasible $\mathbf x$ satisfying
$\mathbb E[f(\mathbf x)]\ge f(\mathbf o)/2$ for an optimal
$\mathbf o$.
\end{theorem}
\begin{proof}
In each branch, the solution remains feasible, and the Top-2
rule never decreases $f$. If $|\operatorname{supp}(\mathbf o)|\le1$, the algorithm
returns a solution of value at least $f(\mathbf o)$.
Assume $|\operatorname{supp}(\mathbf o)|\ge2$ and consider the
branch starting from the maximum-value singleton $\mathbf y\preceq\mathbf o$. Fixing $\mathbf y$
and subtracting the constant $f(\mathbf y)$ preserve
$k$-submodularity, so $g$ is $k$-submodular with $g(\mathbf0)=0$,
although it may be negative.
Each processed candidate counts as one step, whether accepted or
rejected. If the branch ends early, we keep $\mathbf s_t$ and
$\mathbf h_t$ unchanged through step $N=n-1$.

Suppose first that the comparison vectors and weights satisfy
the conditions of Proposition~\ref{prop:potential-analysis}.
Since $r_0=m_0=c(\mathbf h_0)$, we have $\lambda_0=1/2$. Applying the proposition
to $g$ with $\mathbf s_0=\mathbf0$ gives
$\mathbb E[g(\mathbf s_N)]\ge\tfrac12[g(\mathbf h_0)]_+$.
By orthant submodularity,
$f(\mathbf o)-f(\mathbf y\sqcup\mathbf h_0)
\le f(o_z\mathbf1_z)-f(\mathbf0)
\le f(\mathbf y)$,
where the last inequality uses the choice of
$\mathbf y$ and $f(\mathbf0)\ge0$.
Consequently,
\[
\begin{aligned}
\mathbb E[f(\mathbf y\sqcup\mathbf s_N)]
&\ge f(\mathbf y)+\tfrac12[g(\mathbf h_0)]_+\ge\tfrac12\bigl(f(\mathbf y)+f(\mathbf y\sqcup\mathbf h_0)\bigr)
\ge\tfrac12 f(\mathbf o).
\end{aligned}
\]
Since the algorithm returns the best over all branches, it gives a $1/2$-approximation.

We now verify the conditions of Proposition~\ref{prop:potential-analysis}.
Initially $r_0=m_0\ge0$, and the definition of $R_t$ gives
$r_t\le m_t$ at every step. For $0\le r_t\le m_t$, the weight
formula~\eqref{eq:resource-weight} gives $0\le\lambda_t\le1/2$.
We verify below by induction that $r_t\ge0$ whenever \eqref{eq:resource-weight} is used.

The choice of $z$ gives
$c(\mathbf y)+c(\mathbf h_0)+c_z\le1$ and $c_u\le c_z$ for
$u\in\operatorname{supp}(\mathbf h_0)$. Hence, if $u\in R_t$ and
$r_t>0$, then
\[
c(\mathbf y)+c(\mathbf s_t)+c_u
<c(\mathbf y)+c(\mathbf h_0)+c_z\le1.
\]
As $c(\mathbf s_t)$ never decreases, $u$ was also feasible at
every earlier step. Thus it cannot have been rejected and
remains a candidate.

For $t<N$, it remains to verify~\eqref{eq:continuation-condition}
when $D_t>0$ and $G_t<\lambda_tD_t$.
A rejection or a step after termination leaves $\mathbf s_t$,
$R_t$, and the weight unchanged. It has $G_t=0$ and
$D_{t+1}=D_t$, so both inequalities in~\eqref{eq:continuation-condition} hold with $b_t=1$.
Now consider an accepted element $u$.
If $c_u=0$, then $(r_{t+1},m_{t+1})=(r_t,m_t)$ and
$\lambda_{t+1}=\lambda_t$. Lemma~\ref{lem:comparison-interfaces}
gives the first inequality with $b_t=2$ if $u\in R_t$, and
$b_t=\tfrac{1+\sqrt2}{2}$ otherwise. Since
$1-b_t\lambda_t\ge0$ and $G_t\ge0$, the second follows from
$[D_t-b_tG_t]_+\ge D_t-b_tG_t$.

Now suppose $c_u>0$. All zero-cost elements have already been
processed, so $D_t>0$ implies $R_t\ne\varnothing$ and $m_t>0$.
The inequality $G_t<\lambda_tD_t$ gives $\lambda_t>0$ and
hence $r_t>0$. Every $v\in R_t$ is therefore still a candidate,
and the density choice gives
$G_v(\mathbf y\sqcup\mathbf s_t)/c_v\le G_t/c_u$, where
$G_t=G_u(\mathbf y\sqcup\mathbf s_t)$.
Together with Lemma~\ref{lem:comparison-interfaces}, this yields
\begin{equation}\label{eq:knapsack-density}
D_t\le \tfrac{1+\sqrt2}{2}\sum_{v\in R_t}G_v(\mathbf y\sqcup\mathbf s_t)
 \le \frac{\tfrac{1+\sqrt2}{2}m_t}{c_u}G_t.
\end{equation}
If $c_u\ge r_t$, Lemma~\ref{lem:resource-weight} gives
$\lambda_t\le c_u/(\tfrac{1+\sqrt2}{2}m_t)$, and
\eqref{eq:knapsack-density} implies $G_t\ge\lambda_tD_t$,
a contradiction. Hence $0<c_u<r_t$.
If $u\in R_t$, then
$(r_{t+1},m_{t+1})=(r_t-c_u,m_t-c_u)$ and we take $b_t=2$.
If $u\notin R_t$, then
$(r_{t+1},m_{t+1})=(r_t-c_u,m_t)$ and we take
$b_t=\tfrac{1+\sqrt2}{2}$.
Both updates preserve $0\le r_{t+1}\le m_{t+1}$.
This completes the induction showing that $r_t\ge0$ whenever
the weight formula is used. Thus $0\le\lambda_t\le1/2$ at every
step.

Lemma~\ref{lem:comparison-interfaces} gives
$\mathbb E_t[D_{t+1}]\ge D_t-b_tG_t$.
Both $b_t$ and $\lambda_{t+1}$ are determined before the random
label choice at step $t$. Since $1-b_t\lambda_{t+1}\ge0$,
\eqref{eq:knapsack-density} gives
\[
G_t+\lambda_{t+1}[D_t-b_tG_t]_+
\ge\Bigl(\lambda_{t+1}
 +\frac{c_u(1-b_t\lambda_{t+1})}{\tfrac{1+\sqrt2}{2}m_t}\Bigr)D_t
\ge\lambda_tD_t,
\]
where the last inequality follows from
Lemma~\ref{lem:resource-weight} with $r=r_t$ and $m=m_t$.

Finally, at $t=N$, every candidate has been processed.
If $R_N=\varnothing$, then $\mathbf h_N=\mathbf s_N$ and $D_N=0$.
Otherwise, all elements of $R_N$ have positive cost, so $m_N>0$.
If also $\lambda_N>0$, the weight formula gives $r_N=c(\mathbf h_0)-c(\mathbf s_N)>0$.
Thus every element
of \(R_N=\operatorname{supp}(\mathbf h_0)\setminus\operatorname{supp}(\mathbf s_N)\) is still unprocessed, contradicting the fact that all candidates have been processed.
Thus $\lambda_N[D_N]_+=0$, completing the verification of
Proposition~\ref{prop:potential-analysis}.

\end{proof}


		\section{Conclusions and future directions}\label{sec6}
		In this work, we developed an exact-score proportional top-2 randomized framework for non-monotone \(k\)-submodular maximization. The local loss-certificate analysis gives a \((\sqrt2-1)\)-approximation under matroid and knapsack constraints, breaking the previous \(1/3\) barrier. Building on the same local bounds, we use
Bellman-style backward induction to improve
the approximation guarantees to $1/2$ under
cardinality and knapsack constraints.
These guarantees are oracle-optimal for $k=2$
and asymptotically tight as $k\to\infty$.
		
		Several directions remain open. First, it would be interesting to extend the framework to more general constraint families. The local analysis in Section~\ref{sec3} is independent of the feasibility constraint and may be reusable. A natural direction is to combine it with the \(k\)-multilinear-extension framework of Zhou, Huang, and Wang~\cite{zhou2025improved}, or to find Bellman comparison states for multiple knapsacks and matroid intersections.
		
		Second, can the \(\sqrt2-1\) guarantee be improved under a general matroid constraint? Section~\ref{sec3} shows that the original scalar loss certificate cannot achieve \(\eta<\sqrt2\). The cardinality and knapsack refinements show that this limitation need not bound a resource-dependent analysis. A stronger continuation bound, a different comparison invariant, or a different local rule may therefore help in the matroid setting.
		
		Finally, it remains useful to understand the approximation limits for fixed \(k\ge3\), and whether the improved guarantees can be obtained deterministically or with fewer value queries. In particular, reducing the \(O(n^3k^2)\) knapsack complexity while preserving the \(1/2\) guarantee would strengthen the practical value of the framework.

\section*{Acknowledgments}

The authors originally developed the proportional
Top-2 randomized rule and obtained
$(\sqrt2-1)$-approximation guarantees for
matroid and knapsack constraints without AI
assistance. Subsequently, with assistance from
GPT-5.6 Pro, the authors developed a Bellman-style
analysis that improves the guarantees to $1/2$
under cardinality and knapsack constraints while
retaining the underlying Top-2 greedy approach.
The manuscript was drafted and organized by the authors,
with GPT-6 Astra assisting only in language editing
and polishing.
The authors have reviewed all AI-assisted
content and take full responsibility for
the mathematical correctness and final manuscript.

	\end{document}